\documentclass[preprint,11pt,3p]{elsarticle}
\usepackage{graphicx} 
\usepackage{amsmath}
\usepackage{amsthm}
\usepackage{amssymb}
\newcommand{\R}{\mathbb{R}}

\usepackage{changepage}
\usepackage[dvipsnames]{xcolor}
\usepackage{enumerate}
\usepackage{url}
\usepackage{comment}
\usepackage{tikz}
\usetikzlibrary{positioning,arrows.meta,fit,backgrounds}

\definecolor{softblue}{RGB}{220,235,255}
\definecolor{softorange}{RGB}{255,235,205}

\tikzset{
  vertex/.style={
    circle,
    draw=black,
    fill=white,
    minimum size=8mm,
    inner sep=0pt,
    font=\sffamily\small
    },
  	bluevertex/.style={
		vertex,
		fill=blue!25,
		draw=blue!75!black,
		very thick
	},
	redvertex/.style={
		vertex,
		fill=red!25,
		draw=red!75!black,
		very thick
	},
	orangevertex/.style={
		vertex,
		fill=orange!25,
		draw=orange!75!black,
		very thick
	},
    normaledge/.style={
    draw=black,
    line width=0.7pt
   },
   rededge/.style={
    draw=red,
    line width=1.4pt
   },
   arc/.style={
    -{Stealth[length=2.4mm,width=1.8mm]},
    draw=black,
    line width=0.75pt
   },
   cutarc/.style={
    -{Stealth[length=2.4mm,width=1.8mm]},
    draw=red!75!black,
    dashed,
    line width=1pt
   },
  setboxS/.style={
    draw=blue!65!black,
    fill=softblue,
    fill opacity=0.45,
    rounded corners=8pt,
    line width=1pt,
    inner sep=9pt
   },
  setboxT/.style={
    draw=orange!80!black,
    fill=softorange,
    fill opacity=0.55,
    rounded corners=8pt,
    line width=1pt,
    inner sep=9pt
  },
  setlabel/.style={
    font=\sffamily\bfseries,
    text opacity=1
  },
  title/.style={
    font=\sffamily\bfseries
  }
}

\newtheorem{prop}{Proposition}

\newtheorem{lemma}[prop]{Lemma}

\newtheorem{cor}[prop]{Corollary}
\newtheorem{thm}[prop]{Theorem}

\newtheorem{claim}[prop]{Claim}

\newenvironment{claimproof}[1][Proof of Claim]{%
  \begin{proof}[#1]
}{%
  \end{proof}%
}

\newcommand{\fixb}[1]{\textcolor{blue}{#1}}

\newcommand{\calC}{\mathcal{C}}

\usepackage[numbers]{natbib}

\begin{document}

\begin{frontmatter}
\title{The Strong Total Cover Problem}

\author[label1]{Manuel Aprile}
\ead{manuel.aprile@unipd.it}

\author[label2]{Luca Ferrarini}
\ead{ferrarini@lipn.univ-paris13.fr}

\address[label1]{Università degli Studi di Padova,
Mathematics Department, Padova, 35121, Italy}
\address[label2]{Université Sorbonne Paris Nord, CNRS, Laboratoire d’Informatique de Paris Nord, Villetaneuse, 93430, France}

\begin{abstract}
The Strong Total Cover Problem (STCP) is a novel covering problem defined on a graph $G=(V,E)$, where the objective is to find a minimum-cardinality subset $T \subseteq V \cup E$ such that every element of $V \cup E$ is covered. In this setting, a vertex covers itself and all its incident edges, while an edge covers itself together with its two endpoints. The STCP resembles classical covering problems such as Vertex Cover and Edge Cover problems, and naturally arises as the dual of the total matching problem. In this paper, we investigate the STCP from algorithmic, complexity, and polyhedral perspectives.
We first establish close connections between STCP, the Dominating Set Problem, and the Maximum Induced Matching Problem.
These relationships provide structural insights, imply the NP-hardness of STCP, and identify graph classes for which the problem is polynomial-time solvable. We then investigate the weighted version of STCP, showing that it remains NP-hard even on complete graphs with binary weights. On the other hand, we present a 2-approximation algorithm, and a polynomial-time algorithm for complete bipartite graphs.
On the polyhedral side, we introduce the Strong Total Cover Polytope and derive several families of valid inequalities, some of which are shown to define facets. Finally, we provide a complete linear description of the associated polytope for triangular cactus graphs, revealing an interesting contrast with the considerably more intricate facial structure of the dominating set polytope on cactus graphs.
\end{abstract}

\begin{keyword}
Strong Total Cover \sep Total Matching \sep Polyhedral Combinatorics
\sep Approximation Algorithms
\end{keyword}
\end{frontmatter}


\section{Introduction}
Min--max phenomena are a recurrent theme in combinatorial optimization: often,
they arise from the tension between packing pairwise incompatible objects
(e.g., finding a maximum matching in a graph) and covering a set of elements with
given objects (e.g., finding a minimum vertex cover). A standard way to
formalize such relations is through linear programming duality. 
This paper studies a new covering problem that naturally arises from the dual linear program of the total matching problem. The latter generalizes two of the most fundamental packing problems in graphs, that is, the maximum independent set problem and the maximum matching problem. 
Given a simple graph $G=(V,E)$, let $D = V \cup E$ be the set of its \emph{elements}, that is, the set of vertices and edges.
Elements $a, b \in D$ are said to be {\it adjacent} if $a$ and $b$ are adjacent vertices, or incident edges, or $a$ is an edge incident to vertex $b$, or vice versa. A \textit{total matching} is a subset of $D$ whose elements are pairwise non-adjacent, and the \emph{total matching problem} (TMP) asks to find a maximum cardinality total matching in a graph. The TMP is NP-hard \cite{ TotalMatching, Manlove} and it has received some recent attention (see \cite{Polyhedra, TotalMatchingComplete}), also for its relationship with total coloring \cite{behzad1965graphs,vizing1968some}. Indeed, a total coloring is a partition of the elements of a graph into total matchings.

 While the concepts of total matching and total coloring appear in the literature, the natural dual problem, to the best of our knowledge, has not been previously investigated. A \textit{strong total cover} of a graph is a subset of its elements that covers every element, where we say that a vertex covers itself and all incident edges, and an edge covers itself and its two endpoints. The \emph{strong total cover problem} (STCP) asks for a strong total cover of minimum size\footnote{We stress that a different problem called Total Covering Problem, or Mixed Dominating Set Problem, has been studied in recent years \cite{ALAVI1992229, TotalMatching, ErdosTotal,Manlove}. The covering relationship in this problem is more permissive than in STCP, since it allows vertices and edges to cover any adjacent elements, hence we use the word \lq\lq strong'' for our problem, and we omit \lq\lq minimum'' for brevity.}.



One way to arrive at the STCP is through the following natural formulation of
the total matching problem. We recall that, for a graph $G=(V,E)$ and $v\in V$, $\delta(v)$ denotes the set of edges incident to $v$ and $\delta[v]=\delta(v) \cup \{v\}$. 

\begin{align}\label{lp_tm}
\max & \sum_{v\in V} x_v + \sum_{e\in E} y_e & \tag{TM}\\
& x_v + \sum_{e\in \delta(v)} y_e \leq 1  & \forall v \in V \notag\\
& x_v + x_u + y_e \leq 1  & \forall e=(u,v) \in E \notag\\
& 0 \leq x_v\leq 1,\; 0\leq y_e\leq 1& \forall v \in V, \forall e \in E. \notag  
\end{align}

Clearly, imposing integrality constraints on \eqref{lp_tm} one obtains exactly the TMP.
Notice that the constraint matrix of \eqref{lp_tm}  is symmetric. Taking the dual of \eqref{lp_tm} one obtains the linear relaxation of the STCP:

\begin{align}\label{lp_tc}
\min & \sum_{v\in V} x_v + \sum_{e\in E} y_e & \tag{TC} \\
& x_v + \sum_{e\in \delta(v)} y_e \geq 1  & \forall v \in V \notag \\
& x_v + x_u + y_e \geq 1  & \forall e=(u,v) \in E \notag \\
& 0 \leq x_v\leq 1,\; 0\leq y_e\leq 1& \forall v \in V, \forall e \in E. \notag
\end{align}

Let us define \[\alpha_{T}(G):= \max \{|T| : T \mbox{ is a total matching of} G\};\] \[\beta^{*}_T(G) := \min \{ |C| : C \mbox{ is a strong total cover of }G\},\]  or equivalently, the optimal integer values of \eqref{lp_tm}, \eqref{lp_tc} respectively.
Weak duality applied to \eqref{lp_tm} and \eqref{lp_tc} immediately gives, for any graph $G$,
\begin{align}\label{ineq_duality}
\alpha_T(G)\le \beta_T^*(G).\tag{$\star$}
\end{align}


This duality relationship is the starting point of our work. It suggests two
natural questions: when do the packing and covering values coincide, and when
is \eqref{lp_tc} integral? While answering these questions, we noticed that the STCP turns out to have several independent combinatorial
interpretations. We develop these connections first: on the one hand, we show that the STCP is equivalent to the dominating set problem in the class of incidence graphs; on the other hand, we give a Gallai-type identity, showing equivalence to the maximum induced matching problem.  We then investigate the computational complexity of the weighted version of STCP, and finally we initiate a polyhedral study of the
associated covering polytope, giving compact linear descriptions for non-trivial graph classes.
\smallskip 

\noindent {\bf Contributions and organization of the paper.}

In Section \ref{sec:first} we describe the relationship of the STCP with other concepts in graph theory, in particular we relate strong total covers to: \begin{itemize}
    \item dominating sets of incidence graphs: this equivalence implies that the relaxation \eqref{lp_tc} is integral for trees and forests, as is also the case for total matching;
    \item induced matchings: this equivalence shows NP-hardness of STCP;
    \item clique covers of total graphs: this adds to inequality \eqref{ineq_duality} and points to a special class of graphs (triangular cacti) which we investigate later.
\end{itemize}

In Section \ref{sec:complexity}, we give further results on the computational complexity of the weighted version of STCP, where weights are nonnegative. Our results are threefold: we show that the problem is NP-hard, even when restricted to complete graphs with 0/1 weights; we complement this hardness result by giving a 2-approximation algorithm, which is best possible under the Unique Games Conjecture (and $P\neq NP$), also settling the integrality gap of \eqref{lp_tc}; finally, we give a polynomial-time algorithm when the input graph is complete bipartite, a class that has already been object of interest in the context of TMP \cite{TotalMatchingComplete}.

In Section \ref{sec:polyhedral}, we initiate a polyhedral study of the associated Strong Total Cover Polytope, by proposing facet-defining inequalities and by giving complete linear descriptions when the input graph is a tree, a cycle or a class of graphs called triangular cactus graphs\footnote{We recall that a graph is a cactus if its blocks are only cycles and single edges.}. The proof of the latter result is the most technical contribution in the paper, as it requires showing that the proposed
formulation is preserved under $1$-sums, even though this operation may
create facets that are not required for either summand separately. Whether operations like 1-sum or clique-sum preserve formulations or create new complex facets has
been investigated in several contexts, see \cite{Chvatal, aprile2022extended, aprile2024slack}. Our interest in this subclass of cactus graphs is motivated by the
connection with domination in incidence graphs. Indeed, the dominating set polytope of general cactus graphs has a rich and poorly understood facial
structure: in \cite{DominatingSetPolytope}, facets with large coefficients are built through repeated 1-sums.
 Our result on the Strong Total Cover Polytope yields a description of the dominating set polytope of incidence graphs of triangular cactus graphs (which are cactus themselves), adding to the very few classes for which a complete description in the original variable space is known. We remark that the weighted dominating set problem is polynomial-time solvable on cactus graphs, and the polytope admits a compact extended formulation \cite{baiou2014dominating}; since incidence graphs of cactus graphs are cactus, both results carry to STCP. 



\section{First observations and relations with other problems}\label{sec:first}
Throughout the paper, let $G=(V,E)$ be a graph and denote by $D=V\cup E$ its elements.
We start with a simple observation. 
\begin{prop}\label{prop:bounds}
Let $G=(V,E)$ be a graph with $E\neq \emptyset$. Then, the following hold.
\begin{enumerate}
    \item $\beta^*_T(G)\leq 
    |V|-1$;
    \item $\alpha_{T}(G)\leq \beta_{T}^{*}(G)\leq 2\alpha_{T}(G)-1.$
\end{enumerate}
\end{prop}
\begin{proof}
The first part follows as a strong total cover of $G$ can be constructed by choosing any edge $e$ of $G$ and all the vertices that are not endpoints of $e$. 

For the second part, the first inequality is \eqref{ineq_duality} and follows by LP duality as already remarked. For the second inequality, consider first a \lq\lq perfect'' total matching of $G$: this is obtained by taking a maximal matching $M$ of $G$, and the vertices exposed by the matching. This has cardinality $|V|-|M|\geq \frac{|V|}{2}$, implying $\alpha_{T}(G)\geq \frac{|V|}{2}$. This, together with the first part, implies the second inequality.
\end{proof}

Notice that the union of a vertex cover and an edge cover is a strong total cover. Hence, assuming that $G$ has no isolated vertex so that edge covers exist, and denoting by $\beta(G)$ (resp. $\beta'(G)$) the minimum cardinality of a vertex cover (resp. edge cover) of $G$, we have \[\beta^*_T(G) \leq \beta(G)+\beta'(G).\] 
However, a consequence of Gallai's Theorem \cite{Gallai1959} is that $\beta(G)+\beta'(G) \geq |V|$, hence the above bound is strictly weaker than that given by the first part of Proposition \ref{prop:bounds}. On the other hand, it is not hard to see that all the bounds in Proposition \ref{prop:bounds} can be achieved with equality: for instance, a complete graph $K_n$ has $\beta^*_T(K_n)=n-1$ and $\alpha_T(K_n)=\frac{n}2$ for even $n$. We will show two infinite families of graphs satisfying inequality \eqref{ineq_duality} with equality.
\medskip

The questions outlined in the introduction about the relationship between the TMP and the STCP are related to more general questions of packing and covering in hypergraphs and set systems: consider the hypergraph $\mathcal{H}_G$ whose vertices are the elements of a graph $G$, the vertices and the edges of $G$, and whose hyper-edge set is: \[
\mathcal{E}:=\bigl\{\{v\}\cup\delta(v):v\in V\bigr\}
\cup
\bigl\{\{u,v,(u,v)\}:(u,v)\in E\bigr\}.\]
Then, a total matching is a packing of $\mathcal{H}_G$, and a strong total covering is a transversal of $\mathcal{H}_G$. Recall that the incidence graph \(I(G)\) is the bipartite graph with vertex set
\(D=V\cup E\), where \(v\in V\) is adjacent to \(e\in E\) whenever \(v\) is an
endpoint of \(e\). Equivalently, \(I(G)\) is obtained from \(G\) by subdividing
every edge once. We observe that $\mathcal{H}_G$ is the closed neighborhood hypergraph of $I(G)$: indeed, its hyperedges are in bijective correspondence with $\{N[a]: a\in D\}$, where $N[a]$ denotes the closed neighborhood of element $a$ in $I(G)$.  Closed-neighborhood hypergraphs have been studied in several contexts,
including domination, balanced matrices, and algebraic properties of
neighborhood ideals; see, e.g., \cite{Farber1983,NasernejadQureshi2022}.

This observation allows us to derive integrality of the basic linear formulation for the class of forests.

\begin{prop}
    Let $G$ be a forest, then \eqref{lp_tm} and \eqref{lp_tc} are integral. Hence, $\alpha_{T}(G)= \beta_{T}^{*}(G)$.
\end{prop}
\begin{proof}
Let \(A\) be the closed-neighborhood matrix of
\(I(G)\). Thus, the rows and columns of \(A\) are indexed by the elements
\(D=V(G)\cup E(G)\), and
$
(A)_{ab}=1$ if and only if $b\in N[a]$, where again the neighborhood is on $I(G)$.
With this notation, and with $z=(x,y)\in \R^D$, \eqref{lp_tm} can be written as
\[
\{z\in \R^D_{\ge 0}:A z\le \mathbf 1\},
\]
and  \eqref{lp_tc} can be written as 
\[
\{z\in \R^D_{\ge 0}:Az\ge \mathbf 1\}.
\]

Since \(G\) is a forest, its incidence graph \(H=I(G)\) is also a forest.
In particular, \(H\) is strongly chordal. By Farber's characterization of
strongly chordal graphs \cite{Farber1983}, the closed-neighborhood matrix of a strongly chordal
graph is totally balanced. Hence \(A=N[H]\) is totally balanced, and therefore
balanced\footnote{
A \(0,1\)-matrix is \emph{balanced} if it contains no square submatrix of
odd order with exactly two \(1\)'s in every row and every column. It is
\emph{totally balanced} if it contains no square submatrix, of any order at
least \(3\), with exactly two \(1\)'s in every row and every column.
}. 
By a theorem of Fulkerson, Hoffman and Oppenheim \cite{FulkersonHoffmanOppenheim1974} on balanced matrices, we have that both
\eqref{lp_tm} and \eqref{lp_tc} are integral.

Finally, the two linear relaxations are dual linear programs (notice that $A$ is symmetric), hence their integrality, together with strong duality, implies that \(
\alpha_T(G)=\beta_T^*(G)
\).
\end{proof}
 In fact,  the theorem of Fulkerson, Hoffman and Oppenheim implies integrality for all nonnegative
 objectives, and the linear systems of \eqref{lp_tm}, \eqref{lp_tc} are totally dual integral, when $G$ is a forest. 

There are other graphs for which inequality \eqref{ineq_duality} holds with equality, in particular consider a cycle $G=C_{3k}$ for $k\geq 1$, and label its vertices from $1$ to $3k$: the set $\{1, (2,3), 4,\dots, 3k-2, (3k-1,3k)\}$ obtained by alternating vertices and edges is both a total matching and a strong total cover, meaning that is optimal for both problems. One can check that \eqref{lp_tc} is integral for $k=1$, but is fractional for $k=2$, while \eqref{lp_tm} is fractional already for $k=1$ .

\subsection{Equivalence with the dominating set problem}
 A set $S$ of vertices of a graph is a \emph{dominating set} if every vertex not in $S$ has a neighbor in $S$. An equivalent way to reformulate our observation on closed-neighborhood hypergraphs is the following:

\begin{prop}\label{prop:dominating}
A set $S\subset V\cup E$ is a strong total cover of $G$ if and only if it is a dominating set of $I(G)$.    
\end{prop}
\begin{proof}
    We have that $S\subset V\cup E$ is a strong total cover of $G$ if and only if every vertex not in $S$ has an incident edge in $S$, and every edge not in $S$ has an endpoint in $S$. This is exactly the same as $S$ being a dominating set of $I(G)$.
\end{proof}

This implies that the problem of finding a minimum (weight) strong total cover is equivalent to a dominating set problem in the restricted class of incidence graphs (equivalently, bipartite graphs where vertices on one side have all degree 2).

This has several algorithmic consequences for the STCP: first, the dominating set problem admits an asymptotically tight logarithmic approximation algorithm via set cover: however, in Section \ref{sec:complexity} we will see that STCP admits a 2-approximation. 
Moreover, polynomial-time solvability of the Dominating Set Problem on special classes of graphs can carry over to STCP, as the following proposition shows. 

\begin{prop}
The $STCP$ is polynomial-time solvable for series-parallel graphs, cactus graphs, and graphs of bounded treewidth.
\end{prop}
\begin{proof}
Thanks to Proposition~\ref{prop:dominating}, it is enough to
check that the graph classes in the statement are preserved when passing from \(G\) to \(I(G)\), and then use the fact that the dominating set problem is polynomial-time solvable in those classes.

Series-parallel graphs are exactly those graphs that do not have \(K_4\) as a minor, and since edge subdivision
preserves minors, they are closed under edge subdivision. Hence the result
follows from~\cite{SeriesDominating}.

Clearly, if $G$ is a cactus then $I(G)$ is a cactus, and the result follows from \cite{CactusDominating}.

Finally, suppose that \(G\) has treewidth at most \(k\) and let us show that $I(G)$ has treewidth at most $k$. The case $k=1$ is obvious, hence let us consider $k\geq 2$. Let
\((\mathcal T,\{B_t\}_{t\in V(\mathcal T)})\) be a tree-decomposition of
\(G\) of width \(k\). For every edge \(e=(u,v)\in E(G)\), choose a bag \(B_t\)
with \(u,v\in B_t\), and attach to \(t\) a new leaf bag
\(
B_e:=\{u,e,v\},
\)
where \(e\) is now viewed as the subdivision vertex of \(I(G)\). The resulting
collection of bags is a tree-decomposition of \(I(G)\) of width $k$: the edges \((u,e)\) and
\((e,v)\) of \(I(G)\) are covered by \(B_e\), and the bags containing each
original vertex remain connected after adding such leaf bags. Hence
the result follows from the standard dynamic programming algorithm for
dominating set on bounded-treewidth graphs (see, e.g.,
\cite{Bodlaender1988}). 
\end{proof}

Finally, the dominating set problem admits a PTAS on planar graphs (\cite{Baker1994}, see also \cite{KanjPerkovic2004}), and this carries to STCP since, if $G$ is planar, so is any of its subdivisions (and hence $I(G)$).

\subsection{Equivalence with induced matchings}

A matching $M$ of $G$ is \emph{induced} if no edge of $E\setminus M$ has both endpoints covered by $M$.

\begin{prop}\label{prop:induced}
    Let $G=(V,E)$ be a graph and let $\alpha_{ind}(G)$ be the cardinality of the largest induced matching of $G$. Then
    \[
    \beta^*_T(G) + \alpha_{ind}(G) = |V|.
    \]
\end{prop}
\begin{proof} We show that $\beta^*_T(G)=\min\{|M|+|U|: M$ induced matching, $\, U= \mbox{vertices exposed by $M$}\}$. Since the number of vertices exposed by a matching $M$ is $|V|-2|M|$, the quantity $|M|+|U|$ is equal to $|V|-|M|$ and minimized by a maximum induced matching, concluding the proof.

     Let $X\cup Y$ be a strong total cover of $G$, with $Y\subseteq E$, $X\subseteq V$. If $e,f\in Y$ are incident, with $f=(u,v)$ and $u$ an endpoint of $e$. Then replacing $f$ with $v$ yields another strong total cover: indeed, $f$ covers itself, $u, v$, and all three are covered by $v$ and $e$. Iterating this, we obtain that $Y$ is a matching. Assume now that there are edges $e,f\in Y$, $g=(u,v)\in E\setminus Y$ such that $u$ is an endpoint of $e$ and $v$ of $f$. Then $g$ must be covered by one of its endpoints, say $v$, in which case we can replace $f$ with its other endpoint and still have a strong total cover. Hence we obtain that $Y$ is an induced matching. 
     Now, let $v$ be a vertex not covered by $Y$: then $v$ must be in $X$. 
     This proves the $\geq$ inequality.

     For the reverse inequality, consider any induced matching $M$ and let $U$ be the set of exposed vertices. We observe that $M\cup U$ is a strong total cover: indeed, any edge not in $M$ has at least an endpoint that is exposed by $M$ (since $M$ is induced), hence in $U$, and any vertex not in $U$ is an endpoint of an edge of $M$.
\end{proof}

The formula above brings to mind Gallai's identity: $\alpha(G)+\beta(G)=|V|$. Let $H_G$ be the graph obtained as the square of the line graph of $G$ (where the square graph of $H$ is obtained from $H$ by adding an edge between any pair of vertices at distance at most 2). We remark that maximum induced matchings in $G$ exactly correspond to maximum independent sets in $H_G$. Hence, one has \[
\beta^*_T(G) = |V|-\alpha_{ind}(G) = |V|-\alpha(H_G)=|V|-|E|+\beta(H_G),
\]
where we used Gallai's identity and the fact that $|V(H_G)|=|E|$. Equivalently, solving the  STCP on \(G\) is the same, up to the
additive factor \(|V(G)|-|E(G)|\), as solving the vertex cover problem on \(H_G\). We stress that these equivalences do not carry over to the weighted version of STCP.

We conclude the section with some consequences of Proposition \ref{prop:induced} on the complexity of STCP. The following is the most important, and it follows from known facts on the maximum induced matching problem \cite{maxinducedmatching}:
\begin{cor}
    The STCP is NP-hard even on bipartite graphs, and it is polynomial-time solvable for chordal graphs.
\end{cor}
We stress that, while the TMP is NP-hard on bipartite graphs \cite{Manlove}, its complexity status on chordal graphs is an open question \cite{Manlove}. It is not clear whether the polynomial-time solvability of STCP on chordal graphs can help to solve TMP, since, as we now show, the optima of TMP and of STCP can be arbitrarily far apart in chordal graphs.


A graph is a split graph if its vertex set can be partitioned into a clique and a stable set. Split graphs are $2K_2$-free, i.e. they do not contain the graph made by two disjoint edges as an induced subgraph. This implies that $\alpha_{ind}(G)=1$ for any split graph $G$, which, thanks to Proposition \ref{prop:induced}, implies the following.

\begin{prop}\label{prop:threshold}
Let $G$ be a split graph with at least one edge. Then, $\beta^{*}_{T}(G)=|V|-1$.
\end{prop}



\begin{cor}\label{prop:total_cover_complete}
    Let $K_n$ be the complete graph with $n$ vertices, $n\geq 2$. Then
    \(
    \beta^*_T(K_n) = n-1.
    \)
\end{cor}
We remark that $\alpha_T(K_n)=\lceil \frac{n}{2}\rceil$, since a total matching of the complete graph consists of a matching and possibly one vertex.  Thus, even on chordal graphs, the additive gap between \(\alpha_T\) and
\(\beta_T^*\) can be arbitrarily large, and the ratio approaches \(2\).

\subsection{$\beta^*_T$ and the clique cover number of the total graph}

Given a graph $G$, a \emph{clique cover} is a collection of cliques covering all the vertices of $G$, and the \emph{clique cover number} $\bar\chi(G)$ is the smallest size of a clique cover of $G$. It holds that $\bar\chi(G) = \chi(\bar G)$. The \emph{total graph} of $G$, denoted by $T(G)$, is the graph whose vertices are the elements $D$ of $G$ and two elements are adjacent in $T(G)$ if they are adjacent or incident in $G$. Notice that $\alpha_T(G)=\alpha(T(G))=\omega(\overline{ T(G)})\leq \bar\chi(T(G))$.

\begin{prop}\label{prop:cliquecov}
    Let $G$ be a graph. Then $\bar\chi(T(G))\leq \beta^*_T(G)$. Moreover, if $G$ is triangle-free, then equality holds.
\end{prop}
\begin{proof}
    First, we observe that the cliques of $T(G)$ are of four types:
    \begin{enumerate}
        \item  stars of $G$ of the form $u\cup \delta(u)$ for $u\in V$;
        \item cliques of $G$;
        \item a single edge and its end-points (let us call these \emph{edge-cliques});
        \item the three edges in a triangle.
    \end{enumerate}

    Given a strong total cover $X\cup Y$ of $G$, consider the collection of stars centered at $u$ for $u\in X$ and of edge-cliques of $T(G)$ for $e\in Y$. It is easy to check that this is a clique cover of $T(G)$, showing that $\beta^*_T(G)\geq\bar\chi(T(G))$.

    Now, assume that $G$ is triangle-free. Then cliques of type 2 have size at most 2, hence they are contained in cliques of type 3, and there is no clique of type 4. Hence every clique cover can be turned, without increasing in size, into a clique cover made of stars and edge-cliques. It is easy to see that the corresponding vertices and edges form a strong total cover of $G$, proving $\beta^*_T(G)\leq\bar\chi(T(G))$ as desired.

\end{proof}

Figure \ref{fig:net} shows an example where the bound in Proposition \ref{prop:cliquecov} is strict. Interestingly, this \lq net' graph is an example of a graph whose total graph is perfect: indeed, in \cite{RaoRavindra1977PerfectTotalGraphs} it is shown that $T(G)$ is perfect if and only if the blocks of $G$ are either edges or triangles. We call such a graph a \emph{triangular cactus}. While perfection of $T(G)$ gives a clean formulation for the total matching polytope of $G$, it does not say anything on its strong total covering polytope: this will be addressed in Section \ref{sec:polyhedral}.

\begin{figure}
    \centering
    \begin{tikzpicture}[scale=1.2, every node/.style={circle, draw, fill=black, minimum size = 2mm, inner sep=1.5pt}]
    \node (a) at (90:1.2) {};
    \node (b) at (210:1.2) {};
    \node (c) at (330:1.2) {};

    \node (ap) at (90:2.2) {};
    \node (bp) at (210:2.2) {};
    \node (cp) at (330:2.2) {};

    \draw (a) -- (b) -- (c) -- (a);

    \draw (a) -- (ap);
    \draw (b) -- (bp);
    \draw (c) -- (cp);
\end{tikzpicture}
    \caption{If $G$ is the net graph, then $T(G)$ has clique cover number equal to 4: a clique cover is given by the triangle and the three edge-cliques of the leaf edges; on the other hand, one checks that $\beta^*_T(G)=5$, as we need three elements to cover the leaves and two for the triangle.}
    \label{fig:net}
\end{figure}
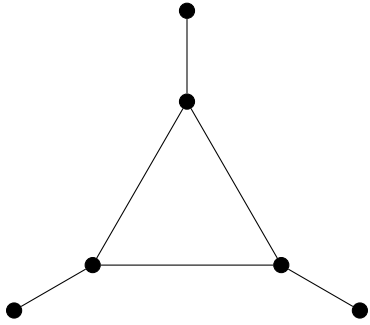


\section{Computational complexity results}\label{sec:complexity}

\subsection{NP-hardness results}
 Let $K_n$ be a complete graph on $n$ vertices. The STCP on $K_n$ is easily solved, see Corollary \ref{prop:total_cover_complete}. However, we next show that the weighted version of the problem (MWSTCP) is NP-hard.

\begin{thm}\label{thm:hardness_complete} The MWSTCP is NP-hard on complete graphs, even if weights are all in $\{0,1\}$.    
\end{thm}
\begin{proof}
    We reduce from the minimum vertex cover problem. Given any graph $G$, consider the complete, weighted graph $K:=(V(G)\cup \{v\}, E(G)\cup E_0)$ obtained from $G$ by adding an extra vertex $v$ and all the missing edges (i.e. $E_0$ is the set of edges between non-adjacent vertices of $G$, and edges between $v$ and vertices of $G$), and giving weight 1 to all vertices and all edges of $E(G)$, and 0 to edges in $E_0$. We now claim that the weight of a minimum weight strong total cover of $K$ is equal to the size of a minimum vertex cover of $G$. 
    Indeed, given a vertex cover $X\subseteq V(G)$, we have that $X\cup E_0$ is a strong total cover of $K$: every edge of $K$ is either in $E_0$ or is in $E(G)$, in which case it is covered by $X$ (since $X$ is a vertex cover of $G$); and every vertex of $K$ is covered by edges in $E_0$ (as those include edges from $v$ to any other vertex). 
    On the other hand, let $X\cup Y$ be a strong total cover of $K$, with $X\subseteq V(K)$, $Y\subseteq E(G)\cup E_0$. Notice that we can assume $E_0\subseteq Y$ as this does not increase the weight of $X\cup Y$. Now, if $e\in E(G)\cap Y$, let us replace $e$ with one of its endpoints: this is still a strong total cover of the same weight, since $e$ is covered by the endpoint and all the vertices are covered by $E_0$. Finally, if the extra vertex $v$ is in $X$, we can safely delete it, since again it is covered by $E_0$. But then we obtain a total cover of form $X\cup E_0$, with $X\subseteq V(G)$, meaning that $X$ is a vertex cover of $G$.
\end{proof}  

It is not hard to check that the reduction above is approximation preserving, which, thanks to well known results \cite{dinur2005hardness, khot2008vertex}, implies the following: 
\begin{cor}\label{cor:hardness_apx} The MWSTCP is NP-hard to approximate within any factor less than 
1.3606. Moreover, assuming the Unique Games Conjecture, it is NP-hard to approximate within any constant factor less than 2.  
Both statements remain true even when the problem is restricted to complete graphs with weights in $\{0,1\}$.  
\end{cor}

\subsection{Polynomial-time Solvability Results}
In contrast to what we have shown before, in the following we prove that the MWSTC is polynomial-time solvable for complete bipartite graphs. Let $K_{m,n}$ be the complete bipartite graph, with $V=A\cup B$ with $|A|=m$ and $|B|=n$. First, notice that the largest induced matching of $K_{m,n}$ consists of only one edge, hence by Proposition \ref{prop:induced} we obtain the following result:
\begin{prop}\label{prop:total_cover_bipartite_complete}
    \[
    \beta^*_T(K_{m,n}) = n+m-1.
    \]
\end{prop}

We now move to the weighted version of STCP. The idea is to reduce the problem to a sequence of minimum-weight cut problems in an auxiliary directed graph, where the goal is to find a bipartition of the vertices into two non-empty sets $(S, V \setminus S)$ such that the total weight of the edges from $S$ to $V \setminus S$ is minimized. 
We solve all these minimum-weight cut problems, together with a greedy algorithm, and return the best solution found.
We initially restrict our attention to the setting in which all weights are nonnegative, and then generalize the analysis to arbitrary weights by means of an appropriate preprocessing phase.

Hence, let us consider a weight function $w: D\rightarrow \R_{\geq 0}$, where $D=A\cup B\cup E$. We need some further definitions.
Let $X\cup Y$ be a strong total cover of $K_{m,n}$, where $X \subseteq A \cup B, Y \subseteq E$, and define $U_A:=A \setminus X$ and $U_B:=B \setminus X$. We say that $X\cup Y$ is \textit{canonical} if $Y$ is exactly the biclique induced by the vertices not included in $X$, that is, if $Y=U_A \times U_B$. Notice that such a set is a minimal strong total cover by construction.
Moreover, we say that $X\cup Y$ is \emph{non-degenerate} if $U_A \neq \emptyset$ and $U_B \neq \emptyset$, that is, $X \neq A$ and $X \neq B$. Otherwise, $X \cup Y$ is said to be \emph{degenerate}.

The following observation ensures that, when looking for a minimum-weight strong total cover, one can focus on canonical ones and degenerate ones. The latter case will be treated separately via a greedy approach.

\begin{claim}\label{cl:bip_canonical}
    Let $X\cup Y$ be an inclusion-wise minimal strong total cover of $K_{m,n}$. Then exactly one of the following holds:\begin{enumerate}
    \item $X\cup Y$ is degenerate and not canonical;
        \item $X\cup Y$ is canonical and non-degenerate;
        \item  $X=A\cup B$, $Y=\emptyset$.
    \end{enumerate} 
\end{claim}
\begin{claimproof}
First, notice that if $X\cup Y$ is canonical and degenerate, then $Y=\emptyset$, leading to case 3. 
We can then assume that $X\cup Y$ is non-degenerate, hence $U_A\neq \emptyset$, $U_B\neq \emptyset$: to conclude the proof, it suffices to show that $X\cup Y$ is canonical. 
Assume by contradiction that $Y\neq U_A\times U_B$. Notice that $Y\supseteq U_A\times U_B$ since all edges must be covered. Hence there is $a\in X$, say without loss of generality $a\in A$, such that $(a,b)\in Y$ for some $b\in B$. By minimality we have $b\not\in X$, hence, for every $a'\in U_A$, $(a',b)\in Y$. But then $b$ is already covered by those edges and $X\cup Y\setminus \{(a,b)\}$ is still a strong total cover, contradicting the minimality of $X\cup Y$.
\end{claimproof}

Thanks to Claim \ref{cl:bip_canonical}, to find a minimum weight strong total cover it suffices to i) find a degenerate strong total cover of minimum weight; ii) find a canonical, non-degenerate strong total cover of minimum weight; iii) check the weight of $A\cup B$. 

To solve problem i), a simple greedy approach suffices: suppose for instance that we want to find a minimum weight cover containing $A$, the other case being symmetrical. Notice that the inclusion of $A$ automatically covers $A$ itself and all the edges. Each $b\in B$ can be covered either by itself or by an incident edge, and this choice does not cover any other vertex in $B$: hence, the minimum weight is achieved simply by including, 
for every $b \in B$, the element of minimum weight in $\{b\}\cup \delta(b)$.

It now remains to solve problem ii).
We construct the directed network $\vec{G}=(N,\vec{E})$ as follows:
\begin{itemize}
    \item $N=\{s,t\} \cup A \cup B$
    \item $\vec{E} =\{(s,a)\, \forall a \in A\} \cup \{(b,t) \,\forall b \in B\} \cup \{(a,b) \,\forall a\in A, b\in B\}.$
\end{itemize}

Fix a pair of vertices $(a_i,b_j) \in A \times B$. We define a capacity function $c = c(i,j):\vec{E}\rightarrow \R_+$ as follows: 
\begin{itemize}
\item $c(s,a_i)= c(b_j,t)=\infty$,
\item $c(s,a)=w(a), \forall a \in A \setminus\{a_i\},$
\item $c(b,t)=w(b), \forall b \in B \setminus \{b_j\}$
\item $c(a,b)=w(a,b), \forall e=(a,b) \in E.$
\end{itemize}

Let $(S,T)$ be a cut of $\vec{G}$ of minimum capacity, with $T=N\setminus S$, $s \in S, t \in T$. Clearly, $a_i\in S$, $b_j\in T$, as otherwise the cut would have infinite capacity. Define $X=(A\cap T)\cup (B\cap S)$, and  $Y = U_A \times U_B$, where $U_A = A \cap S=A\setminus X$, $ U_B = B \cap T=B\setminus X$ as before. See Figure \ref{fig:s,t-cut} for the details of this construction.

\begin{figure}[h]
\begin{center}
	\begin{tikzpicture}[node distance=1.5cm and 2.4cm, scale = 0.9]
		
		
		\node[redvertex] (a1) at (0,0) {$a_1$};
		\node[redvertex, below=of a1] (a2) {$a_2$};
		\node[vertex, below=of a2] (a3) {$a_3$};
		
		\node[vertex, right=of a1] (b1) {$b_1$};
		\node[redvertex, below=of b1] (b2) {$b_2$};
		\node[redvertex, below=of b2] (b3) {$b_3$};
		
		\foreach \a in {a1,a2,a3} {
			\foreach \b in {b1,b2,b3} {
				\draw[normaledge] (\a) -- (\b);
			}
		}
		
		\draw[rededge] (a3) -- (b1);
		
		\begin{scope}[xshift=8.5cm]
			
			\node[bluevertex] (s) at (-2.2,-1.5) {$s$};
			
			\node[orangevertex] (da1) at (0,0) {$a_1$};
			\node[orangevertex, below=of da1] (da2) {$a_2$};
			\node[bluevertex, below=of da2] (da3) {$a_3$};
			
			\node[orangevertex] (db1) at (2.6,0) {$b_1$};
			\node[bluevertex, below=of db1] (db2) {$b_2$};
			\node[bluevertex, below=of db2] (db3) {$b_3$};
			
			\node[orangevertex] (t) at (4.8,-1.5) {$t$};
			
			
			\draw[arc] (s) -- (da3);
			\draw[arc] (da1) -- (db1);
			\draw[arc] (da2) -- (db1);
			\draw[arc] (da3) -- (db2);
			\draw[arc] (da3) -- (db3);
			\draw[arc] (db1) -- (t);
			
			\draw[cutarc] (s) -- (da1);
			\draw[cutarc] (s) -- (da2);
			\draw[cutarc] (da3) -- (db1);
			\draw[cutarc] (db2) -- (t);
			\draw[cutarc] (db3) -- (t);
			\draw[arc] (da1) -- (db2);
			\draw[arc] (da2) -- (db3);
			\draw[arc] (da2) -- (db2);
			\draw[arc] (da1) -- (db3);
		\end{scope}
		
	\end{tikzpicture}
  \caption{A minimal strong total cover $C=\{a_1,a_2,b_2,b_3,(a_3,b_1)\}$ of $K_{3,3}$ and the corresponding minimum cut with $S=\{s,a_3,b_3,b_2\}$ and $\delta(S)=\{(s,a_1),(s,a_2),(a_3,b_1),(b_2,t),(b_3,t)\}$.}\label{fig:s,t-cut}
\end{center}
\end{figure}
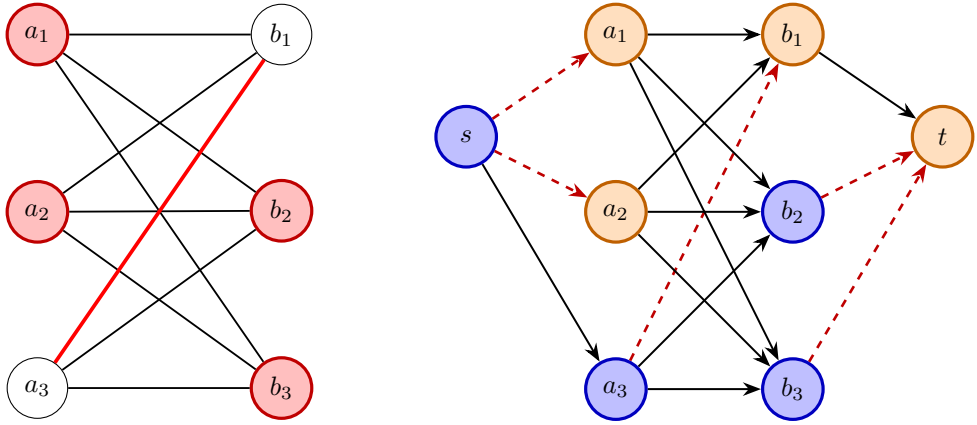

\begin{claim}\label{cl:bip_cut}
   The set $X\cup Y$ defined in this way is a canonical, non-degenerate strong total cover, and $w(X\cup Y)=c(\delta(S))$.
\end{claim}
\begin{claimproof}
    That $X\cup Y$ is a canonical strong total cover follows by definition, and it is non-degenerate since $a_i\in U_A$ and  $b_j\in U_B$. To prove the last part of the claim, it remains to show that the elements of $X\cup Y$ are in bijective correspondence with arcs in $\delta(S)$. Indeed, the arcs $\delta(S)$ are either of type $(s,a)$ for $a\in T$, in which case $a\in X$, or of type $(b,t)$ for $b\in S$, in which case $b\in X$, or of type $(a,b)$ for $a\in S$, $b\in T$, in which case $(a,b)\in Y$, and by construction every element of $X\cup Y$ corresponds to an arc of $\delta(S)$. The conclusion follows from the definition of $c$.
\end{claimproof}

From the proof it is clear that every canonical, non-degenerate strong total cover $X\cup Y$ with $a_i\in U_A$, $b_j\in U_B$ arises from a cut of $\vec{G}$ of finite capacity. Hence, by solving a minimum $s-t$ cut problem on $\vec{G}$ for each choice of $a_i\in A, b_j\in B$, we find a canonical, non-degenerate strong total cover of minimum weight as desired.

This concludes the algorithm for the case of nonnegative weights. We now show how to deal with negative weights. 

\begin{thm}\label{thm:w_comp_bip}
The MWSTC problem is polynomial-time solvable on complete bipartite graphs.
\end{thm}
\begin{proof}
Let $G=(A \cup B,E,w)$ be a complete bipartite graph with arbitrary weights $w: D \longrightarrow \mathbb{R}$. The crucial observation is that there is an optimal solution of MWSTCP containing every element with non-positive weight, and negative weight elements are contained in every optimal solution. Hence, consider a new weight function $\bar{w}$ such that for any element $k \in D$, $\bar{w}(k):=w(k)$ if $w(k) \geq 0$ and $\bar{w}(k):=0$ if $w(k) <  0$.  A strong total cover of $G$ that is optimal with respect to $w$ can be obtained by finding one that is optimal with respect to $\bar w$ and then adding all negative weight elements. Hence, we can run the algorithm described above for $G$ with weights $\bar w$, and then add to the solution obtained all the elements $k\in D$
 with $w(k)<0$.
\end{proof}

\subsection{Approximation algorithms}
We summarize here what is known about approximation algorithms, starting with the unweighted case.

\begin{lemma}\label{lem:trivial_2apx}
The TMP admits a $\frac{1}{2}$-approximation algorithm, and the STCP admits a $2$-approximation algorithm. 
\end{lemma}
\begin{proof}
    Both parts are a direct consequence of Proposition \ref{prop:bounds}, that gives:
    \[\frac{n}2\leq \alpha_T(G)\leq \beta^*_T(G)\leq n-1.\]
    A total matching satisfying the approximation guarantee is given by any maximal matching together with the exposed vertices; a strong total cover satisfying the approximation guarantee is obtained by picking all the vertices, or, any edge and all the vertices apart from its endpoints.  
\end{proof}

It is natural to ask whether similar approximation guarantees also hold for the weighted version of the problem. We first remark that one cannot hope for such a result for the total matching problem (unless $P=NP$), because, if all the edge weights are zero, then the problem reduces to the maximum independent set problem, which is notoriously hard to approximate \cite{zuckerman2006linear}. A similar issue holds if the vertex weights are much larger than the edge weights. 
On the other hand, we now show that the weighted STCP admits a 2-approximation algorithm.

Let $(G,w)$ be a strong total cover instance, where $w:E\cup V\rightarrow \R_{\geq 0}$ are weights on vertices and edges of $G$. Notice that we restrict to nonnegative weights, as it is customary to do when proving multiplicative approximation guarantees.
Consider the weighted version of \eqref{lp_tc}, and its dual (formally, the dual is taken after dropping the $\leq 1$ constraints):
\begin{align*}
\min\quad & \sum_{v\in V} w(v)\,x_v + \sum_{e\in E} w(e)\,y_e\\
\text{s.t.}\quad &
x_v + \sum_{e\in\delta(v)} y_e  \ge\ 1  & \forall v\in V,\tag{WTC}\\
& y_e + x_u + x_v  \ge\ 1 & \forall e=(u,v)\in E,\\
& 0 \leq x_v\leq 1, & \forall v \in V, \\
& 0 \leq y_e\leq 1, & \forall e \in E.
\end{align*}
\begin{align*}\label{lp_D}
\max\quad & \sum_{v\in V} a_v + \sum_{e\in E} b_e\\
\text{s.t.}\quad &
a_v + \sum_{e\in\delta(v)} b_e \ \le\ w(v) & \forall v\in V,\tag{D}\\
& a_u + a_v + b_e \ \le\ w(e) & \forall e=(u,v)\in E,\\
& a_v,b_e\ge 0.
\end{align*}

Then, consider the modified dual $D_0$ where the variables $a_v$ are set to 0 for any $v$. 
Let $(\mathbf{0}, \bar b)$ be an optimal solution of $D_0$, notice that it is a feasible solution of $D$. Let $X_0\cup Y_0\subseteq V\cup E$ be the set formed in this way: \begin{itemize}
    \item for any $v\in V$, add $v$ to $X_0$ iff $\bar b(\delta(v)) = w(v)$;
    \item for any $e\in E$, add $e$ to $Y_0$ if i) $\bar b_e=w(e)$; ii) at least one endpoint of $e$ is \emph{not} in $X_0$.
\end{itemize}

\begin{claim}\label{cl:apx_cover}
    $X_0\cup Y_0$ covers every edge: for $e\in E$, either $e\in Y_0$ or $e$ has an endpoint in $X_0$.
\end{claim}
\begin{claimproof}
    By optimality of $\bar b$, every variable $b_e$ appears in at least one tight edge constraint or vertex constraint.
\end{claimproof}


However, the set $X_0\cup Y_0$ does not necessarily cover all vertices. Let $S \subseteq V$ be the set of vertices not covered by $X_0 \cup Y_0$.
By definition of $X_0 \cup Y_0$, a vertex $v \in V$ belongs to $S$ if and only if $\bar b(\delta(v))<w(v)$, and $\bar b(e)<w(e)$ for $e\in \delta(v)$.
Notice that, thanks to Claim \ref{cl:apx_cover}, $S$ is an independent set. We now construct a new feasible solution $(\bar{a},\bar{b})$ of $D$ by raising the $a$ variables, while leaving $\bar b$ unchanged (and keeping $\bar{a}_u = 0$ for all $u \in V \setminus S$).
Specifically, since $S$ is an independent set, each edge $e \in E$ has at most one endpoint in $S$ , so we can independently set 
\[
\bar a_v = \min \{w(v)-\bar b(\delta(v)),  (w(e)-\bar b_e): e\in \delta(v)\},
\]
We then construct new sets $X_1 \subseteq S$ and $Y_1 \subseteq E$ by selecting for each $v \in S$, exactly one constraint that becomes tight:
\begin{enumerate}
    \item if $\bar{a}_v = w(v)- \bar{b}(\delta(v))$, we add $v$ to $X_1$;
    \item otherwise, we pick a single edge $e=(u,v) \in \delta(v)$ such that $\bar{a}_v+\bar{a}_u + \bar{b}_e = w(e)$ (notice that $e \notin Y_0$ and $\bar{a}_u = 0$ since $ u \notin S$), and add $e$ to $Y_1$.
\end{enumerate}
Let $X = X_0 \cup X_1$ and $Y = Y_0 \cup Y_1$. By construction, this process adds for each $v \in S$, either $v$ itself to $X_1$ or a single edge incident to $v$ to $Y_1$, implying that $X \cup Y$ is a strong total cover. Moreover, for every edge $e=(u,v) \in E$, if $u, v \in X$, then $e \notin Y$: indeed, if $e \in Y_0$ then at least one endpoint is outside $X_0$ (and hence outside $X$, as vertices covered by $Y_0$ are not in $S$), while if $e \in Y_1$, then its endpoint in $S$ was not added to $X_1$.
\begin{claim}
    $(\bar a, \bar b)$ is a feasible solution of $D$.
\end{claim}
\begin{claimproof}
   Indeed, we have $\bar a_v+\bar b(\delta(v))\le w(v)$ for every $v\in S$ and
$\bar a_u+\bar a_v+\bar b_e\le w(e)$ for every $e=(u,v)$; all remaining
constraints follow from the feasibility of $\bar b$ for $D_0$.
\end{claimproof}

\begin{claim}
    We have that $w(X\cup Y)\leq 2\cdot OPT$.
\end{claim}
\begin{claimproof}

    By construction of $X_0\cup Y_0$, we have \[w(X_0)=\sum_{v\in X_0} \bar b(\delta(v)); \qquad w(Y_0) = \bar b(Y_0).\]
    Moreover, since each $v \in S \setminus X_1$ is incident to exactly one edge $e \in Y_1$, and its other endpoint $u \not\in S$ has $\bar{a}_u = 0$, we have:
    \[w(X_1)=\bar a(X_1) + \sum_{v\in X_1} \bar b(\delta(v)); \qquad w(Y_1) = \sum_{(u,v)\in Y_1} (\bar{a}_u+\bar a_v) +  \bar b(Y_1) = \bar{a}(S \setminus X_1)+\bar{b}(Y_1),\]
    Hence, we have:
    \[w(X_1\cup Y_1) = \bar a(S) + \sum_{v\in X_1} \bar b(\delta(v)) + \bar b(Y_1).\]
    Since $X $ (resp. $Y$) is the disjoint union of $X_0$ and $X_1$ (resp. $Y_0$ and $Y_1$), summing the two identities yields:
    \[
    w(X\cup Y) \leq \bar a(S) + \sum_{v\in X} \bar b(\delta(v)) + \bar b(Y) \leq \bar{a}(V)+2\bar{b}(E).
    \]
  where the inequality holds because for every edge $e = (u,v) \in E$, we cannot simultaneously have $u, v \in X$ and $e \in Y$, so $\bar{b}_e$ appears at most twice in $\sum_{v \in X} \bar{b}(\delta(v)) + \bar{b}(Y)$. 
Finally, since $(\bar{a}, \bar{b})$ is a feasible dual solution, $\bar{a}(V) + 2\bar{b}(E) \le 2\bigl(\bar{a}(V) + \bar{b}(E)\bigr) \le 2 \cdot \text{OPT}$.
This concludes the proof.
\end{claimproof}

Hence we obtain:
\begin{thm}
    \label{thm:2apx}
    The weighted STCP with nonnegative weights admits a 2-approximation algorithm.
\end{thm}

We conclude this section by noting that the proof also yields an upper bound of 2 on the integrality gap of (WTC), which can be shown to be tight similarly as done for the integrality gap of the basic relaxation of the vertex cover problem (see also the proof of Theorem \ref{thm:hardness_complete}).

\begin{prop}\label{prop:gap2}
The integrality gap of (WTC) is
equal to \(2\).
\end{prop}

\begin{proof}
The upper bound follows directly from the proof of
Theorem~\ref{thm:2apx}. Indeed, the analysis shows that the weight of the strong total cover returned by the algorithm is at most twice the value of a feasible dual solution.
Hence its
weight is at most twice the optimal value of the (WTC), and therefore the
integer optimum is at most twice the LP optimum.

We now show that the factor $2$ is tight. Let $n\ge 2$, and consider the
complete graph $K_{n+1}$ with a distinguished vertex $v$. Assign weight 1 to every vertex of $K_{n+1}$, weight $0$ to every edge of $\delta(v)$, and weight 1 to every remaining edge.

Notice that, in any integer solution, we can assume without loss of generality that the zero-weight edges are part of the solution, hence all vertices are covered, and it only remains to cover the weight $1$ edges, that is, the edges of $K_{n+1} \setminus v$.
It is immediate to see that the integer optimum is obtained by selecting a vertex cover of $K_{n+1}\setminus v$, hence its value is $n-1$.

On the other hand, the optimum value of the linear relaxation is at most $n/2$, as witnessed by the feasible solution defined by  $x_u=\frac12, \forall u \in V(K_{n+1}) \setminus \{v\}$ and $y_e=1$ for every $ e\in \delta(v)$, and $y_e=0$ for any other edge. 
Hence,
\[
\frac{OPT_{IP}}{OPT_{LP}} \geq \frac{n-1}{n/2}=2-\frac{2}{n}
\]
which tends to $2$ as $n \rightarrow \infty$. This concludes the proof.
\end{proof}

\section{Polyhedral Results}\label{sec:polyhedral}
In this section we tackle the STCP from a polyhedral perspective. 
Given a strong total cover $C \subseteq D$, the corresponding characteristic vector $\chi^{C}\in \{0,1\}^D$ is defined as follows:
\begin{equation*}
  \chi^{C}_a:=\left\{
  \begin{array}{@{}ll@{}}
    1 & \text{if}\ a \in C  \\
    0 & \text{if}\ a \in D\setminus C.
  \end{array}\right.
\end{equation*}
Then, the Strong Total Cover Polytope $P_{TC}(G)$ is defined as follows:
\[
P_{TC}(G):=\hbox{conv}\{\chi^{C}: \text{$C$ is a strong total cover of $G$} \}.
\]
Observe that, if the graph $G$ contains isolated vertices, then the  $P_{TC}(G)$ fails to be full-dimensional. 
\begin{prop}\label{vertex_constr}
Let $G=(V,E)$ be a connected graph with $|V|\geq 2$. Then, $\dim(P_{TC}(G))=|V|+|E|$, that is $P_{TC}(G)$ is full-dimensional.
\end{prop}
\begin{proof}
Consider the set of incidence vectors $\chi^{D},\chi^{D \setminus a}, \forall a \in D$. It is easy to check that they form a set of $|D|+1$ affinely independent points. 
\end{proof}
We restate the basic inequalities that are clearly valid for $P_{TC}(G)$.
\begin{align}
\label{eq:vertex_edge} &x_v + \sum\limits_{e \in \delta(v)}y_e\geq 1,\qquad &\forall v \in V \\
\label{eq:edge_vertex}\qquad &x_v+x_w+y_e \geq 1,\qquad &\forall e=(v,w) \in E \\
\label{eq:box1}\qquad &z_a \geq 0,\qquad  &\forall a \in D \\
\label{eq:box2}\qquad &z_a \leq 1,\qquad   &\forall a \in D
\end{align}
We refer to \eqref{eq:vertex_edge} as the \textit{vertex inequalities}, to \eqref{eq:edge_vertex} as the \textit{edge inequalities}, and to \eqref{eq:box1} together with \eqref{eq:box2} as the \textit{box inequalities}. 
First, we establish that \eqref{eq:vertex_edge} and \eqref{eq:box2} are always facet-defining for $P_{TC}(G)$.
\begin{prop}\label{edge_constr}
Let $G=(V,E)$ be a connected graph with $|V|\geq 2$. Then inequalities \eqref{eq:vertex_edge} and \eqref{eq:box2} are facet-defining for $P_{TC}(G)$.
\end{prop}
\begin{proof}
Let $ a \in D $. Consider the strong total covers $C_0: = D$ and, for every $b \in D \setminus \{a\}$, $C_b: = D \setminus \{b\}$, and the corresponding incidence vectors. 
 The latter are $|D|$ and they lie on the face induced by $z_a \leq 1$. To check affine independence, it suffices to subtract $\chi^{C_0}$ from $\chi^{C_b}$ for $b \in D \setminus \{a\}$, obtaining the independent canonical vectors $-\mathbf{e}_b$. This shows that Inequalities \eqref{eq:box2} are facet-defining.
 Now, let $v \in V$ and consider the face $F$ induced by vertex-inequality associated to $v$. Consider the following strong total covers: define $C_a := (D \setminus \delta[v]) \cup \{a\}, \forall a \in \delta[v]$; for $a=v$ define $C^{v}_b:=C_v \setminus \{b\}, \forall b \in D \setminus\delta[v] $ if $b$ is not a neighbor of $v$ of degree 1; define $C^{b}:=(D \setminus \delta[v]) \setminus \{b\}\cup (v,b)$ for every degree-1 neighbor of $v$.  It is easy to check that the corresponding $|D|$ incidence vectors of the latter set of strong total covers lie on the face $F$. 
  To check affine
independence, subtract $\chi^{C_v}$. For $c\in\delta(v)$ the
resulting vectors are $\mathbf{e}_c-\mathbf{e}_v$; for $b\in D\setminus\delta[v]$ that is not a
degree-$1$ neighbor of $v$, they are $-\mathbf{e}_b$; while for every degree-$1$
neighbor $b$ they are
$
-\mathbf{e}_b+\mathbf{e}_{(v,b)}-\mathbf{e}_v
$.
Subtracting the vector $\mathbf{e}_{(v,b)}-\mathbf{e}_v$ from the latter gives
$-\mathbf{e}_b$. Hence the resulting vectors are linearly independent, and therefore
the vertex inequality defines a facet.

\end{proof}
In the following Proposition we show that under suitable conditions  the edge-inequalities and \eqref{eq:box1} are facet-defining for $P_{TC}(G)$.

\begin{prop}\label{prop:edge_facet}
Let $G$ be a connected graph and let $e=(v,w) \in E$. Then, the edge-inequality \eqref{eq:edge_vertex} and $y_e \geq 0 $ are facet-defining for $P_{TC}(G)$ if and only if $\deg(v),\deg(w) \geq 2$. Moreover, $x_u \geq 0$ for $u \in V$ is facet-defining if and only if $\deg(u) \geq 2$.
\end{prop} 
\begin{proof}
First, suppose that $e$ has the end-point $v$ with $\deg(v)<2$, hence $\deg(v)=1$ since $v$ is not isolated. Then $x_v+x_w+y_e \geq 1$ can be written as the sum of the vertex-inequality $x_v + y_e \geq 1$ and $x_w \geq 0$, hence it cannot define a facet. Similarly, $y_e \geq 0$ can be obtained as the sum of $x_v + y_e \geq 1$ and $-x_v \geq -1$. 
Analogously, if \(u\in V\) has \(\deg(u)=1\), and \(f\) is the unique edge
incident with \(u\), then \(x_u\ge0\) is obtained as the sum of
$
x_u+y_f\ge 1
$
and
$
-y_f\ge -1,
$
which is the upper bound \(y_f\le 1\).

We now prove the converse statements.
First, suppose $\deg(v),\deg(w) \geq 2$. To show that inequality $y_e\geq 0$ is facet-defining,
 let $C_0 := D\setminus \{e\}$ and, for every $a \in D \setminus \{e\}$, let $C_a := D \setminus \{e,a\}$. Notice that the assumption on the degree of $v,w$ ensures that $C_0, C_a$ are strong total covers. The corresponding incidence vectors are $|D|$, and they lie on the face induced by $z_e \geq 0$. To check affine independence, it suffices to subtract $\chi^{C_0}$ from $\chi^{C_a}$ for  $a \in D \setminus \{e\}$, obtaining the independent canonical vectors $-\mathbf{e}_a$. 

To show that inequality \eqref{eq:edge_vertex} is facet-defining,
 let: $C_e := D\setminus \{v,w\}$, $C_v := D\setminus \{e,v\}$, 
 $C_w := D\setminus \{e,w\}$,
 $C_a := D \setminus \{v,w,a\}$ for every $a \in D \setminus \{e,v,w\}$. Again, the assumption on the degree of $v,w$ ensures that all of these are strong total covers. The corresponding incidence vectors are $|D|$, and they lie on the face induced by $x_v+x_w+y_e \geq 1$. To check affine independence, it suffices to subtract $\chi^{C_e}$ from all the others, obtaining the canonical vectors $-\mathbf{e}_a$ for $a \in D \setminus \{e,v,w\}$, and vectors $\mathbf{e}_w-\mathbf{e}_e$, $\mathbf{e}_v-\mathbf{e}_e$, that are linearly independent.

Finally, let \(u\in V\) with \(\deg(u)\ge2\). To show that \(x_u\ge0\) is
facet-defining, let
\(
C_0:=D\setminus\{u\}
\)
and, for every \(a\in D\setminus\{u\}\), let
\(
C_a:=D\setminus\{u,a\}.
\)
Again, it is easy to check that these are all strong total covers. The corresponding incidence vectors are $|D|$ and lie on the face \(x_u=0\). Subtracting
\(\chi^{C_0}\) from the vectors \(\chi^{C_a}\), \(a\in D\setminus\{u\}\), gives
the linearly independent vectors
\(
-\mathbf e_a,a\in D\setminus\{u\}.
\)
Hence \(x_u\ge0\) defines a facet. This completes the proof.
 
\end{proof}
A natural question arises from edge inequalities \eqref{eq:edge_vertex}, as they appear to belong to a particular family of rank-type constraints: $z(D(H))\geq \beta^{*}_T(H)$ where $H$ is an induced subgraph of $G$.  
 Notice that Inequalities \eqref{eq:edge_vertex} correspond to the case where the induced graph is just an edge, i.e. $K_2$. We now give a natural generalization of this. We stress that such rank constraints are not necessarily valid, since the vertices of a subgraph may be covered by edges that do not belong to the subgraph itself.  Take for instance the graph $G$ constructed from the complete bipartite graph $K_{2,2}$ by adding a pendant edge to each vertex. 

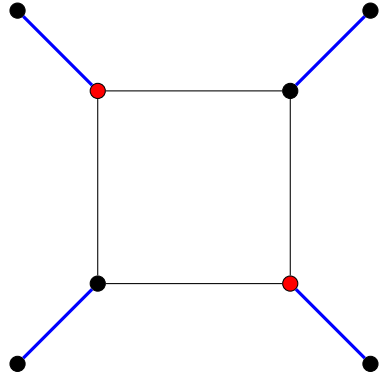
\begin{figure}[h]
    \centering
    \begin{tikzpicture}[scale=1.2, every node/.style={circle, draw, fill=black, minimum size = 2mm, inner sep=1.5pt}]
    \node[fill = red] (a) at (135:1.5) {};
    \node (b) at (45:1.5) {};
    \node[fill = red] (c) at (-45:1.5) {};
    \node (d) at (-135:1.5) {};

    \node (ap) at (135:2.75) {};
    \node (bp) at (45:2.75) {};
    \node (cp) at (-45:2.75) {};
    \node (dp) at (-135:2.75) {};

    \draw (a) -- (b) -- (c) -- (d) -- (a);
    \draw[blue, very thick] (a) -- (ap);
    \draw[blue, very thick] (b) -- (bp);
    \draw[blue, very thick] (c) -- (cp);
    \draw[blue, very thick] (d) -- (dp);
\end{tikzpicture}
    \caption{The graph $G$ and the strong total cover composed by the blue edges and red nodes.}
    \label{fig:C4_net}
\end{figure}

The rank-constraint inequality associated with $K_{2,2}$ is 
\[
\sum\limits_{v \in V(K_{2,2})}x_v + \sum\limits_{e \in E(K_{2,2})}y_e \geq4,
\]
see Proposition \ref{prop:total_cover_bipartite_complete}.
 However, this inequality is not valid for $G$, since it is violated by the strong total cover shown in Figure \ref{fig:C4_net}. 

In contrast, we show that rank constraints corresponding to cliques are valid, and under some natural condition they are even facet-defining. 
\begin{thm}\label{clique_ineq}
Consider a connected graph $G$ and let $K_{n}$ be an induced complete subgraph of $G$, with $n\geq 2$. Then, the $n$-clique inequality
\begin{align}\label{ineq:sub_clique}
\sum\limits_{v \in V(K_{n})}x_v + \sum\limits_{e \in E(K_{n})}y_e \geq n-1
\end{align}
is a valid inequality for $P_{TC}(G)$. Moreover, if each $v \in V(K_n)$ has at least one neighbor in $V \setminus V(K_{n})$, then it is facet-defining for $P_{TC}(G)$.
\end{thm}
\begin{proof} 
Let $C$ be a strong total cover of $G$,
and let $p$ be the number of vertices of $K_n$ not in $C$. We can assume $p>1$ otherwise the inequality is clearly valid for $\chi^C$. Then, since $C$ must contain all the edges of $K_n$ whose endpoints are not in $C$, we have:
\[
|C\cap D(K_n)|\geq n-p+{p \choose 2}= n-1+\frac{(p-1)(p-2)}2\geq n-1. 
\]
This concludes the first part of the proof. Now, with the additional hypothesis we prove that the $n$-clique inequality is facet-defining.
Let $F$ be the face induced by the inequality \eqref{ineq:sub_clique} and suppose by contradiction that $F \subseteq \tilde{F}=\{z \in P_{TC}(G)\mid \lambda^{T}z = \lambda_{0} \}$. We show that there exists $a \in \mathbb{R}$ such that $\lambda_{z} =a,\forall z \in D(K_{n})$ and $\lambda_{z}=0,\forall z \in D \setminus D(K_{n})$. Define the sets $C_u:=(V(K_n) \setminus \{u\}) \cup ( D \setminus D(K_{n}))$ for $u \in V(K_{n})$ and $C_e:=(V(K_{n}) \setminus\{u,w\}) \cup \{e\} \cup (D \setminus D(K_{n}))$ for $e=(u,w) \in E(K_{n}) $. It is easy to check that they form strong total covers (in particular, in $C_u$ the vertex $u$ is covered by an incident edge in $E \setminus E(K_n)$ which exists since $u$ has a neighbor in $V \setminus V(K_{n})$). 
Since $|C_{e} \cap D(K_n)|=|C_{u} \cap D(K_n)|=n-1$, we have $\chi^{C_u},\chi^{C_e} \in F \subseteq \Tilde{F}$, 
in particular for $e=(u,w)\in E(K_n)$, we have $\lambda^{T}\chi^{C_{u}}=\lambda^{T}\chi^{C_e}$. This implies that $\lambda_{w}= \lambda_{e}$, since the two expressions differ by exactly one element. By repeating this procedure for all the elements of $K_n$ we obtain $\lambda_{z} = a \in \mathbb{R}, \forall z \in D(K_{n})$. Now, define $C^{e}_a:= C_e \setminus \{a\}$ 
for $ a\in D \setminus D(K_{n})$ and $e=(v,w) \in E(K_n)$. It turns out that $\chi^{C^{e}_a} \in F$ and in particular $\lambda^{T}\chi^{C^{e}_a} = \lambda^{T}\chi^{C_e}=\lambda_0$, this implies that $\lambda_a = 0$. By applying the same argument we have that $\lambda_a=0,\forall a \in D\setminus D(K_{n})$. Moreover, $\chi^{C_u} \in F$ and hence $\lambda_0 = (n-1)a$. This concludes the proof.
\end{proof}

Notice that the additional hypothesis of the Theorem \eqref{clique_ineq} is necessary.
Indeed, let $G = K_n$ and $v \in V(K_{n})$ and consider the induced subgraph by deleting $v$, denote it as $K_{n-1}$. Then, the $n$-clique inequality associated to $G$ can always be written as the sum of
\[
\sum\limits_{v \in V(K_{n-1})}x_v + \sum\limits_{e \in E(K_{n-1})}y_e \geq n-2
\]
and 
\[
x_v + \sum\limits_{e \in \delta(v)}y_e \geq 1
\]
which gives rise to the $n$-clique inequality.

\paragraph{Separation of $r$-clique inequalities} We now study the complexity of separating $r$-clique inequalities, and show that the separation problem is NP-complete.
Let \textit{$r$-Clique Separation Problem} ($r$-CSP) denote the separation problem of a $r$-clique inequality, where $r \geq 3$ is part of the input. Let $L_{TC}(G)$ be the linear relaxation of $P_{TC}(G)$ given by the vertex, edge and box inequalities.
The \textit{r-CSP} reads as follows:

\begin{itemize}
    \item \textbf{Instance:} A graph $G'=(V',E')$ and a point $(x^{*},y^{*}) \in \mathbb{Q}^{|V'|+|E'|}$, with $r \in \mathbb{N}, r \geq 3$.
    \item \textbf{Decide:} If there exists a violated $r$-clique inequality of $L_{TC}(G')$ of size $r$, that is, if there is a clique $V(K) \subseteq V(G')$ such that $|V(K)|=r$ and 
    \[
    \sum\limits_{v \in V(K)}x^{*}_v + \sum\limits_{e \in E(K)}y^{*}_e <r-1
    \]
\end{itemize}

Consider the $k$-Clique Problem ($k$-CP) which is the following:
\begin{itemize}
    \item \textbf{Instance:} A graph $G=(V,E)$ and an integer $k$.
    \item  \textbf{Decide:} If $G$ contains a complete subgraph $S \subseteq V(G)$ such that $|S| = k$.
\end{itemize}
The latter problem was shown in \cite{karp1975} to be NP-complete. 
We reduce from \textit{$k$-CP}. Membership in NP is immediate, since a clique of size $r$ serves as a polynomial certificate and the corresponding inequality can be checked in polynomial time. Now, let $(H,r)$ be an arbitrary instance of the $k$-CP, where $H=(V_H,E_H)$ is a graph and $r \geq 3$ is an integer. We construct an instance $(G',x^{*},y^{*},r)$ of the \textit{r-CSP} as follows. Let $G'$ be the graph constructed from $H$ by introducing a pendant edge for each vertex of $H$. Specifically, let $V'=V_H \cup \{u_v \mid v \in V_H\}$ and $ E'= E_H \cup \{(v,u_v) \mid \forall v \in V_H\}$. Now, define a feasible point $z^{*}=(x^{*},y^{*})$ as follows. Set $x^{*}_v = 1/2, \forall v \in V_H$ and $y^{*}_e = 0,\forall e \in E_H$, and $x^{*}_{u_v}=y^{*}_ {(v,u_v)}=1/2$ for every pendant vertex and edge. It is easy to check that $z^{*} \in L_{TC}(G')$ since all the inequalities of the linear relaxation are satisfied. Now, suppose $H$ contains a clique $S \subseteq V_H$ of size $r$. Since $H$ is a subgraph of $G'$, $S$ is also a clique of size $r$, thus the corresponding $r$-clique inequality is well-defined. Now, evaluating the associated inequality at the point $(x^{*},y^{*})$ gives
\[
\sum\limits_{v \in V(S)}x^{*}_v+\sum\limits_{e \in E(S)}y^{*}_e = r/2 < r- 1
\]
Since $r \geq 3$, $S$ defines a violated $r$-clique inequality in $G'$. For the other direction, suppose that there exists a violated $r$-clique inequality in $G'$. By definition, $G'$ contains a clique of size $r$. This must be a clique of $H$: indeed, every new vertex \(u_v\) is pendant and cannot be contained in a $r$-clique with \(r\ge3\). Therefore, the \textit{$r$-CSP} is NP-complete when $r$ is part of the input.

\subsection{Complete Polyhedral Characterizations}
In this section, we provide complete linear descriptions of $P_{TC}(G)$ for certain graph classes. Notice that, if a graph $G$ has components $G_1$, $G_2$, then strong total covers of $G$ are exactly the disjoint unions of strong total covers of $G_1$ and $G_2$, hence $P_{TC}(G)$ can be obtained by juxtaposing the descriptions of $P_{TC}(G_1)$ and $P_{TC}(G_2)$ (i.e. taking their cartesian product). Hence, in the following we restrict ourselves to connected graphs. In light of Proposition \ref{dom_cover} and of the fact that the dominating set polytope has a complicated facial structure already for cactus graphs, see \cite{DominatingSetPolytope}, we restrict ourselves to simple classes. On the positive side, we are able to derive the first complete linear characterization for the triangular cactus graphs.
First, we recall our result on forests and trees. 
\begin{thm}
Let $G$ be a tree. Then, $P_{TC}(G)$ is completely described by the inequalities \eqref{eq:vertex_edge},\eqref{eq:edge_vertex} and box-inequalities \eqref{eq:box1}--\eqref{eq:box2}.
\end{thm}

Now, let $C_n $ be a cycle graph with vertex set $V(C_n)=\{v_1, v_2, \dots, v_n\}$ and edge set $E(C_n)=\{e_1,e_2, \dots, e_n\}$ where $e_i=(v_i,v_{i+1})$ and index $n+1$ is identified with $1$. Notice that the incidence graph of $C_n$ is a cycle with $2n$ vertices. In particular, the equivalence between the strong total cover problem and the dominating set problem on the incidence graph yields the following Lemma. Let $P_{D}(G)$ be the dominating set polytope, that is, the convex hull of characteristic vectors of all dominating sets of a graph $G$.
\begin{lemma}\label{dom_cover}
Let $G$ be a graph. Then, $P_{TC}(G)=P_{D}(I(G))$.
\end{lemma}
Hence, we have $P_{TC}(C_n)=P_{D}(C_{2n})$ by Lemma \ref{dom_cover}.  The dominating set polytope of a graph $H$ can be modeled as the set covering polytope $Q(A)=\mbox{conv}\{ x \in \{0,1\}^{|V(H)|}: Ax \geq 1\}$, where $A$ represents the closed neighborhood matrix of $H$. In the specific case of $H=C_{2n}$, $A$ is the circulant matrix having exactly $3$ consecutive ones in each row. Following the notation of \cite{MinorInequalities}, we denote this matrix as $C^{3}_{2n}$, yielding $P_{TC}(C_n)=Q(C^{3}_{2n})$.
It was shown in \cite{DominatingSetPolytope} from the viewpoint of the
dominating set polytope, and independently in \cite{FacialStructure} from the
viewpoint of the set covering polyhedron, that the \textit{rank constraint} $\sum\limits_{i \in V}x_i \geq \lceil \frac{n}{3}\rceil$ is always a valid inequality for $Q(C^{3}_{n})$, and it is a facet if and only if $n \neq 0 \mod 3$. This implies the following consequence:
\begin{prop}
Let $C_n$ be a cycle of $n$ vertices. Then:
\begin{align}\label{rank:constraint}
\sum\limits_{v \in V(C_n)}x_v + \sum\limits_{e \in E(C_n)}y_e \geq \bigg\lceil \frac{2n}{3}\bigg\rceil
\end{align}
is facet-defining for $P_{TC}(C_n)$ if and only if $n \neq 0 \mod 3$.
\end{prop}
In \cite{DominatingSetPolytope}, the authors introduced another family of facet-defining inequalities. Subsequently, the authors \cite{MinorInequalities} showed that these inequalities are equivalent to the so-called \emph{minor inequalities}. Let $W=\{a_1,a_2, \dots,a_l\} \subset V \cup E$ such that the number of elements between $a_i$ and $a_{i+1}$ (where indices are taken$\mod l$) is a multiple of 3 that we denote by $3k_i$.
Then, the \textit{minor inequalities} read as follows:
\begin{align}\label{eq:minor_ineq}
2\sum\limits_{i \in W}z_i + \sum\limits_{i \notin W}z_i \geq \sum\limits_{i=1}^{l}k_i+ \big \lceil\frac{l}{2}\big \rceil
\end{align}
The authors in \cite{DominatingSetPolytope} proved that \textit{minor inequalities}, basic-facet-defining inequalities and rank constraint fully describe the $P_{TC}(C_n)$ when the underlying graph is a cycle of length $n$. This allows us to derive the following consequence
\begin{thm}
Let $C_{n}$ be a cycle of $n$ vertices. Then, $P_{TC}(C_n)$ is completely characterized by \eqref{eq:vertex_edge}--\eqref{eq:box2},\eqref{rank:constraint} and \eqref{eq:minor_ineq}.
\end{thm}

We now turn our attention to triangular cactus graphs. For our purposes, given a graph $G$ denote with $\mathcal{T}(G)$ the set of all triangles of $G$.
We recall that the total graph $T(G)$ of a triangular cactus graph is perfect, as shown in \cite{RaoRavindra1977PerfectTotalGraphs}. Moreover, since stable sets of $T(G)$ are in one-to-one correspondence with total matchings of $G$, we have
$
P_{\mathrm{stable}}(T(G)) = P_{TM}(G)
$,
where $P_{stable}$ denotes the stable set polytope.
By a classical theorem of Chvátal \cite{Chvatal}, the stable set polytope of a perfect graph is completely described by its clique inequalities which now we introduce. Let $\mathcal{K}$ be the set of all maximal cliques of $T(G)$. Then, the clique inequalities read as follows:
\begin{align}\label{ineq:cliq}
\sum_{v\in K} x_v \le 1, \quad \forall K \in \mathcal{K}
\end{align}
Notice that all maximal cliques of $T(G)$ fall into precisely  four categories: star-cliques induced by $v \cup \delta(v),\forall v \in V(G)$; the edge-cliques $\{u,v,e\}$ for each $e=(u,v) \in E(G)$; vertex-triangles cliques $\{u,v,w\}$ induced by the three vertices of each triangle in $G$ and edge-triangle cliques $\{e_1,e_2,e_3\}$ induced by the three edges of each triangle in $G$. 
Hence, we derive a complete description of $P_{TM}(G)$. First we recall the basic inequalities describing the $P_{TM}(G)$:
\begin{align}
\label{basic_vertex} &x_v + \sum\limits_{e \in \delta(v)}y_e \leq 1, \qquad &\forall v \in V,\\
\label{basic_edge}&x_v+x_w+y_e \leq 1, \qquad & \forall e=(v,w) \in E,\\
\label{basic_box}&z_a \geq 0, \qquad & \forall a \in D.
\end{align}
\begin{prop}
Let $G$ be a triangular cactus graph. 
Then, $P_{TM}(G)$ is completely described by \eqref{basic_vertex}--\eqref{basic_box}, $\sum\limits_{v \in T}x_v \leq 1, \forall T \in \mathcal{T}(G)$ and $\sum\limits_{e \in E(T)}y_e \leq 1, \forall T \in \mathcal{T}(G)$.
\end{prop}
Although the polytope $P_{TM}(G)$ admits a complete description when the underlying graph is a triangular cactus, the same does not hold for the strong total cover polytope. To obtain a complete characterization, we need an additional family of inequalities, which naturally arise as the covering analogue of the clique inequalities.
We show next that, for any such graph $G$, $P_{TC}(G)$ is described by vertex inequalities \eqref{eq:vertex_edge}, edge inequalities \eqref{eq:edge_vertex}, box inequalities \eqref{eq:box1}--\eqref{eq:box2}, and clique inequalities \eqref{ineq:sub_clique} for the triangles of $G$. We adopt a decomposition approach: our cactus $G$, apart from trivial cases, can always be decomposed as a 1-sum of two triangular cactus graphs $G_1$ and $G_2$. We will show that, if our formulation described above is integral for $G_1$ and $G_2$, then it is integral for $G$. This is non-trivial, since the 1-sum operation, simple as it may appear, can introduce new facet-defining inequalities that were not needed for any of the two pieces. For instance, one can check that if $G$ is the net graph in Figure \ref{fig:net} and $G_2$ is a single edge, hence $G_1$ is  a triangle with two leaves, then both $P_{TC}(G_1)$ and $P_{TC}(G_2)$ can be described by \eqref{eq:vertex_edge}-\eqref{eq:box2} only, while the clique inequality is facet-defining for $P_{TC}(G)$.

The proof relies on the following idea. Let $z\in \R^{D(G)}$ satisfying our formulation. Using the integrality of the formulations for $G_1$ and $G_2$, we show that $z\in P_{TC}(G)$. The point $z$ can be projected onto $\R^{D(G_1)}$ and $\R^{D(G_2)}$. However, these projected points might fail to satisfy the formulations for $G_1$, $G_2$, since they may violate the inequality associated with the vertex 
involved in the $1$-sum. To overcome this issue, we suitably modify the projected points so that they belong to $P_{TC}(G_1)$ and $P_{TC}(G_2)$. By integrality, each of these points can then be expressed as a convex combination of strong total covers of the corresponding graph $G_1$ and $G_2$. 
The final step consists of combining these strong total covers into strong total covers of $G$, thereby expressing the original point $z$ as a convex combination.
The latter is the most technical step; it repeatedly uses similar elementary exchange
operations. Suppose that a graph is obtained as a 1-sum \(H=H_1\oplus H_2\).
If two covers behave well at the common vertex, then one may exchange their
\(H_2\)-parts and still obtain two covers. This operation preserves the sum of
the two incidence vectors, and therefore can be used to move from a convex combination to an equivalent one with better properties. We record this observation separately.

\begin{lemma}\label{lem:swap}
    Let $H=H_1\oplus H_2$, such that $h$ is the vertex involved in the 1-sum, and let $A, B$ be strong total covers of $H$. Suppose that one of the following holds:
    \begin{enumerate}
        \item $h\in A\cap B$;
        \item $h\in A\setminus B$, and $A\setminus D(H_2)=A\cap D(H_1)\setminus\{h\}$ is a strong total cover of $H_1$.
    \end{enumerate}
    Then $A'=A\setminus D(H_2)\cup (B\cap D(H_2))$ and $B'=B\setminus D(H_2) \cup (A \cap D(H_2))$ are strong total covers of $H$. Moreover, $\chi^A+\chi^B=\chi^{A'}+\chi^{B'}$.
\end{lemma}
\begin{proof}
The equation among incidence vectors immediately follows from the definitions, hence to conclude the proof it suffices to check that, in both cases, all elements of $H_1$ and $H_2$ are covered by $A', B'$. In the first case this is immediate, since all elements of $H_i$ are covered by both $A\cap D(H_i)$ and $B\cap D(H_i)$ (with $h$ belonging to both), for $i=1,2$. Let us focus on the second case: the assumption implies that all the elements of $H_1$ (including $h$) are covered by $A\setminus D(H_2)$, and similarly all elements of $H_2$ (apart from $h$) are covered by $B\cap D(H_2)$ since $h\not \in B$, implying that $A'$ is a cover. $B'$ is clearly a cover since $B\setminus D(H_2)$ covers all elements of $H_1$ apart from $h$, and $A\cap D(H_2)$ covers all elements of $H_2$ including $h$.
    \end{proof}

We are now ready to state our result.

\begin{thm}
Let $G$ be a triangular cactus graph. Then, $P_{TC}(G)$ is completely described by the basic inequalities \eqref{eq:vertex_edge}--\eqref{eq:box2} and $3$-clique inequalities $\sum\limits_{v \in V(T)}x_v+\sum\limits_{e \in E(T)}y_e\ge 2$
for every triangle $T \in \mathcal{T}(G)$.
\end{thm}

\begin{proof} For the sake of simplicity we use the term \lq\lq cover'' instead of \lq\lq strong total cover''. 
Recall that the blocks of $G$ are edges and triangles; we write $G=G_1\oplus G_2$ to indicate that $G$ is the 1-sum of triangular cacti $G_1$, $G_2$, i.e. $G$ is obtained from them by identifying a vertex.
Let $z=(x,y)\in \R^{D(G)}$ be a point satisfying \eqref{eq:vertex_edge}-\eqref{eq:box2} and all 3-clique inequalities, we want to show that $z\in P_{TC}(G)$, i.e. that $z$ is a convex combination of incidence vectors of (strong total) covers of $G$: since our relaxation is a rational polyhedron, it suffices to show the claim when $z$ is a rational vector. We will proceed by induction on the number of blocks of $G$. The base case is immediate, since the basic relaxation \eqref{eq:vertex_edge}-\eqref{eq:box2} describes $P_{TC}(G)$ exactly when $G$ is a single edge or a triangle. For the induction step, let $G=G_1\oplus G_2$ be any 1-sum decomposition of $G$. Let $z^{(1)}$ be the restriction of $z$ to the elements of $G_1$. Notice that $z^{(1)}$ satisfies all edge inequalities \eqref{eq:edge_vertex}, box constraints \eqref{eq:box2}, and all vertex inequalities \eqref{eq:vertex_edge} apart from the one relative to the vertex involved in the 1-sum, which we call $r$. Hence, let $s_1=\max(0,1-x_r-y(\delta_{G_1}(r)))$ be the slack of this inequality, and define $\bar z^{(1)}$ from $z^{(1)}$ by replacing the $x_r$ coordinate with $x_r+s_1$. Now, $\bar z^{(1)}$ satisfies \eqref{eq:vertex_edge}-\eqref{eq:box2} and all $3$-clique inequalities for $G_1$, hence by induction it belongs to $P_{TC}(G_1)$, i.e. there exist covers $C_1,\dots, C_k$ of $G_1$ such that $\bar z^{(1)}$ is a convex combination of incidence vectors $\chi^{C_1},\dots, \chi^{C_k}$. By the rationality of $z$, hence of $\bar z^{(1)}$, by multiplying by a sufficiently large integer $N$ we have that there exists a multiset $\calC^{(1)}$ of  covers of $G_1$ such that \[
\bar z^{(1)} = \frac{1}{N} \sum_{C\in \calC^{(1)}} \chi^C.
\]
 Following the same idea, let $z^{(2)}$ be the projection of $z$ on the elements of $G_2$ and let $s_2=\max(0,1-x_r-y(\delta_{G_2}(r)))$, and $\bar z^{(2)}\in P_{TC}(G_2)$ be defined similarly as above, and again there exists a multiset $\calC^{(2)}$ of covers of $G_2$ such that \(
\bar z^{(2)} = \frac{1}{N} \sum_{C\in \calC^{(2)}} \chi^C,
\)
where we use the same $N$ as above since we can assume that $|\calC^{(1)}|=|\calC^{(2)}|=N$. Indeed, if they have different sizes 
$N_1$, $N_2$, we take multiple copies of each set until they have the same size equal to the least common multiple of $N_1, N_2$.
Notice that if $s_1=s_2=0$, i.e. if $z^{(i)}=\bar z^{(i)}\in P_{TC}(G_i)$ for $i=1,2$, then it is easy to match covers in $\calC^{(1)}$ with covers in $\calC^{(2)}$ and construct a convex combination for $z$, concluding the proof. This is indeed a special case of the proof that follows. Hence from now on we can assume that, without loss of generality, $s_1>0$.

The latter implies that the vertex inequality \eqref{eq:vertex_edge} at $r$ is tight at $\bar z^{(1)}$, hence it is tight for any $\chi^C$, that is, any $C$ contains either $r$ or exactly one of the edges of $\delta(r)$.

Let us define the following partition of $\calC^{(1)}$: \begin{itemize}
    \item $\calC^{(1)}_0$: covers not containing $r$;
    \item $\calC^{(1)}_s$: covers containing $r$, and also containing all neighbors of $r$ in $G_1$.
    \item $\calC^{(1)}_r$: all the others, i.e. covers containing $r$ and not every neighbor of $r$ in $G_1$.
\end{itemize}
 
And similarly define $\calC^{(2)}_0$, $\calC^{(2)}_s$, $\calC^{(2)}_r$. 
The following claim states that, if $\calC^{(1)}_s$, $\calC^{(2)}_s$ are large enough, then $z$ is a convex combination of incidence vectors of covers of $G$, concluding the proof of the theorem.
\begin{claim}\label{cl:cactus1} If $|\calC^{(1)}_s|\geq Ns_1$ and $|\calC^{(2)}_s|\geq Ns_2$, then there is a multiset $\calC$ of covers of $G$  with $|\calC|= N$ such that 
 \[z= \frac{1}{N} \sum_{C\in \calC} \chi^C.\]
\end{claim}

\begin{claimproof}
We show that $\mathcal{C}^{1}$ and $\mathcal{C}^{2}$, that have both size $N$, can be matched exactly to form the desired multiset $\mathcal{C}$. Notice that, for any $C_1 \in \mathcal{C}^{1}$ and $C_2 \in \mathcal{C}^{2}$, their union $C_1\cup C_2$ is a cover of $G$. However, to satisfy the equation of the claim we need to remove $r$ from some of such pairs: indeed, $\mathcal{C}^{i}$ contains exactly  $N\bar z^{(i)}_r$ sets containing $r$, for $i=1,2$, whereas $\calC$ needs to have $Nx_r$ of them, and we recall that $x_r=\bar z^{(i)}_r-s_i$ for $i=1,2$.

To ensure that the pairing can be made, let us first show that $s_1 +s_2\leq 1-x_r$. Indeed, if  $s_1=0$ or $s_2=0$, the inequality holds trivially. Otherwise, we must have $x_r+y(\delta_{G_i}(r))<1$ and $s_i=1-x_r-y(\delta_{G_i}(r))$ for $i=1,2$, then \[s_1+s_2 = 1-x_r-y(\delta(r)) +1-x_r \leq 1-x_r.\] 

Hence we have \[|\calC^{(1)}_0|=N(1-\bar z^{(1)}_r)=N(1-x_r-s_1)\geq Ns_2,\] and similarly one obtains $|\calC^{(2)}_0|\geq Ns_1$. We can then define the following pairs, which we group into four families: 

\begin{itemize}
\item $\mathcal{C}_{s_2}$ made by exactly $Ns_2$ sets of form $C_1\cup C_2-r$ with $C_1\in \calC^{(1)}_0$, $C_2\in \calC^{(2)}_s$;
 \item $\mathcal{C}_{s_1}$ made by exactly $Ns_1$ sets of form  $C_1\cup C_2-r$ with  $C_1\in \calC^{(1)}_s$, $C_2\in \calC^{(2)}_0$;
 \item this leaves $|\calC^{(1)}_0|-Ns_2 = N(1-x_r-s_1-s_2)$ sets in $\calC^{(1)}_0$ and the same number of sets, i.e. $|\calC^{(2)}_0|-Ns_1 = N(1-x_r-s_1-s_2)$, in $\calC^{(2)}_0$, hence we can match them exactly in pairs $C_1\cup C_2$, $C_1\in \calC^{(1)}_0$, $C_2 \in \calC^{(2)}_0$ to form $\mathcal{C}_{out}$;
 \item This also leaves $|\calC^{(1)}_s\cup \calC^{(1)}_r|-Ns_1=N x_r$ sets in $\calC^{(1)}$, and the same number in $\calC^{(2)}$, hence we can match them exactly in pairs $C_1\cup C_2$, $C_1\in \calC^{(1)}_s\cup \calC^{(1)}_r$, $C_2 \in \calC^{(2)}_s\cup \calC^{(2)}_r$ to form $\mathcal{C}_{in}$.
 \end{itemize}

Notice that, thanks to the definition of sets in $\calC_s^{(i)}$, the sets defined above are all covers of $G$. Let $\calC=\calC_{s_1}\cup \calC_{s_2} \cup \calC_{out}\cup \calC_{in}$.
It suffices to check that $\calC$ satisfies the claimed equation. For each component of $z$ besides $x_r$, this immediately follows from the definitions of $\calC^{(1)}$, $\calC^{(2)}$, and the fact that each set in each of them appears as part of exactly one pair of $\calC$.   For $x_r$, notice that exactly $Nx_r$ sets in $\calC$ contain $r$, completing the proof.

\end{claimproof}

 We remark that, if we are in the case where $s_2=0$, the hypothesis of the above Claim for $\calC^{(2)}_s$ is trivially satisfied. But, for at least one of $i=1,2$, we might have that $|\calC^{(i)}_s|$ is too small: the following Claim shows that in this case we can rearrange the covers in $\calC^{(i)}$ in order to obtain a large enough $\calC^{(i)}_s$.

\begin{claim}\label{cl:cactus2} The sets $\calC^{(i)}_0, \calC^{(i)}_r, \calC^{(i)}_s$ can always be modified so as to satisfy the assumptions of Claim~\ref{cl:cactus1}.
\end{claim}

\begin{claimproof}
    We fix $i=1$, the case $i=2$ follows the same argument. 
We prove the claim by contradiction.
First, notice that $|\calC^{(1)}_s|<Ns_1$ implies that  $|\calC^{(1)}_r|>Nx_r$ since $|C^{(1)}_{s}|+|C^{(1)}_{r}|=N(x_r+s_1)$ (where $x_r+s_1$ is the $r$-coordinate of $\bar{z}^{(1)}$). 
In particular, $\calC^{(1)}_r  \neq \emptyset$. Moreover, 
we can assume $s_1=1-x_r-y(\delta_{G_1}(r))>0$.

Recall that $\bar z^{(1)}$ can be expressed as a convex combination of covers of $\calC^{(1)}$. However, there may exist multiple admissible choices for $\calC^{(1)}$. In particular, if the elements of two covers of $\calC^{(1)}$ can be rearranged in the manner described in Lemma \ref{lem:swap}, then this configuration yields another valid multiset. 
 We assume by contradiction that, among all possible choices for such a valid multiset of size $N$, every choice gives $|\calC_s^{(1)}|<Ns_1$. Then, we choose $\calC^{(1)}$ that maximizes $|\calC_s^{(1)}|$, and, subject to that, maximizes the largest intersection of a cover of $\calC^{(1)}_r$ with $N_{G_1}(r)$, i.e. $\max_{C\in \calC^{(1)}_r} |C\cap N_{G_1}(r)|$ is as large as possible. Fix $\bar C$ achieving the maximum. We show next that $\calC^{(1)}$ can be modified into another valid multiset, while either increasing  $|\calC^{(1)}_s|$ or leaving it unchanged but increasing the intersection $|\bar C\cap N_{G_1}(r)|$, a contradiction. 

We now proceed with the proof. Notice that $N_{G_1}(r)$ consists of two types of vertices: 
\begin{enumerate}[i)]
    \item vertices of $G_1$ that do not lie on a triangle with $r$ (these are cut-vertices or leaves of $G_1$);
    \item vertices that lie on a triangle with $r$.
\end{enumerate}

 First, assume that there is a vertex $w\in N_{G_1}(r)$ of the type i), such that $w\not \in \bar C$. Notice that, by definition, $\bar z^{(1)}$ satisfies at equality the vertex inequality at $r$. This implies that all the covers $C\in \calC^{(1)}$ satisfy $|C\cap (\{r\}\cup \delta_{G_1}(r))|=1$. On the other hand, the edge inequality for $(r,w)$ is satisfied with slack at least $s_1$ by $\bar z^{(1)}$ (i.e. $\bar{x}^{(1)}_r+\bar{x}^{(1)}_w +\bar{y}^{(1)}_{(r,w)} -1 \geq s_1$). Indeed, the inequality is already satisfied by $z^{1}$ as restriction of $z$, and in $\bar{z}^{(1)}$ the variable $\bar{x}_r$ has been increased by $s_1$.
 This implies that there are sets in $\calC^{(1)}$ with positive slack, i.e. containing at least two elements among $r, w, (r,w)$. However, notice that no set in $\mathcal{C}^{(1)}$ can simultaneously contain both $r$ and $(r,w)$. Hence, any set in $\mathcal{C}^{(1)}$ with positive slack must contain $w$ together with exactly one of the elements $r$ or $(r,w)$: the slack must indeed be 1. This implies that there are at least $Ns_1$ such sets with positive slack, hence there is at least one of them, which we call $C'$, that does not belong to $\calC^{(1)}_s$. Assume first that $r, w\in C'$. 
 Let $B(r,w)$ denote the \emph{branch rooted at $r$ containing $w$}, namely, the subgraph of $G_1$ induced by $r$ and all vertices whose paths to $r$ pass through $w$.\footnote{Since we are going to use the concept of rooted branches multiple times, we stress that, since $G$ is a cactus, a vertex can be connected to $r$ through many paths, but either they all contain $w$ or none of them does. }
 Then, apply Lemma \ref{lem:swap} to $G_1$ with $h=r$, $H_2 =B(r,w)$, $A=\bar C$ and $B=C'$. Let $A', B'$ be the covers obtained by the Lemma. 
 Replacing $\bar C$ and $C'$ in $\calC^{(1)}$ with $A'$ and $B'$ yields another valid multiset of covers. Since $|A'\cap N_{G_1}(r)|> |\bar C \cap N_{G_1}(r)|$, this contradicts 
 the choice of $\calC^{(1)}$, as it produces either a larger value of $|\bar C\cap N_{G_1}(r)|$, or a larger value of $|\calC^{(1)}_s|$.
 Hence, we must have that $(r,w), w\in C'$. Let $\overline{B(w,r)}$ be the \textit{branch rooted at $w$ not containing $r$}, namely, the subgraph of $G_1$ induced by $w$ and all the vertices separated from $r$ by $w$. Then, apply again Lemma \ref{lem:swap} to $G_1$ with $h=w$, $H_2 = \overline{B(w,r)}$, 
 $A=C'$, and $B=\bar C$. We again reach a contradiction.

Hence, we must have that any vertex in $N_{G_1}(r)\setminus \bar C$ is of type ii). 
Since $\bar C\not\in \calC^{(1)}_s$, there must exist such a vertex. Thus, there is a triangle $T= \{r, u, v\}$ such that $|\{u,v\}\cap \bar C|<2$. Suppose that there exists another cover $C\in \calC^{(1)}_r$ with $u, v\in C$. 
Let $B(r,\{u,v\})$ denote the \textit{branch rooted at $r$ containing $u$ and $v$}, namely, the subgraph of $G_1$
induced by $r$ and all vertices whose paths to $r$ pass through either $u$ or $v$.
Applying Lemma \ref{lem:swap} to $G_1$ with $h=r$, $H_2 = B(r,\{u,v\})$, $A=\bar C$, and $B=C$, we obtain new covers $A'$ and $B'$. Replacing $C$ and $\bar{C}$ by $A'$ and $B'$ we obtain again a contradiction, since this choice yields another valid multiset in which $|A' \cap N_{G_1}(r)|>|\bar{C} \cap N_{G_1}(r)|.$ 
Therefore, we have $|\{u,v\}\cap C|<2$ for all $C\in \calC^{(1)}_r$. Henceforth, we may assume without loss of generality that $u\not\in \bar C$. 

Now, arguing as above, the edge inequality associated with $(r, u)$ is satisfied by $\bar z^{(1)}$ with slack at least $s_1$. Moreover, since every cover $C'\in \calC^{(1)}$ satisfies $|C'\cap \{r,(r,u)\}|\leq 1$, any cover $C'$ with positive slack for this inequality must contain $u$ together with exactly one of $r$ and $(r,u)$ with slack $1$. 
Hence, there are at least $Ns_1$ such covers. Since, by assumption $|\calC^{(1)}_s|<Ns_1$, there exists one such cover, say $C_u$, that does not belong to $\calC^{(1)}_s$. Applying the same argument to the edge inequality associated with $(r,v)$, we obtain a cover $C_v \notin \calC^{(1)}_s$ and such that $v\in C_v$, and $|C_v\cap \{r, (r,v)\}|=1$.

Suppose that $(r,u)\in C_u$. Since $C_u$ must cover the edge $(r,v)$, and neither $r$ nor $(r,v)$ belong to $C_u$, it follows that $v\in C_u$. Then, apply Lemma \ref{lem:swap} to $G_1$ with $h=u$, $H_2=\overline{B(u,r)}$
, $A=C_u$, and $B=\bar C$. Obtain a new cover $B'$ in $\calC^{(1)}_s$ or in $\calC^{(1)}_r$ with a larger intersection with $N_{G_1}(r)$, a contradiction in both cases. The same argument applies when $r\in C_u$ and $(u,v)\in C_u$. Therefore, we may assume that $r\in C_u$ and $(r,u), (u,v), v\not\in C_u$. 
From the previous argument we can also derive the following consequences:
 \begin{itemize}
     \item Suppose that there exists a cover $C'\in \calC^{(1)}$ such that $(u, v)\in C'$ and $C'\setminus \{(u,v)\}$ is still a cover (we say that $C'$ \emph{can afford} to lose the edge). Then, we may first move the edge $(u,v)$ from $C'$ to $C_u$, after which we obtain the same contradiction. Hence from now on we also assume that no cover $C'\in \calC^{(1)}$ can afford to lose edge $(u,v)$.
     \item We must have $v\in\bar C$. Indeed, otherwise
     $\bar C\cap\{u,v\}=\emptyset$, and applying
      Lemma~\ref{lem:swap} to $G_1$ with $h=r$, $H_2=B(r,\{u,v\})$, 
       $A=C_u$, and $B=\bar C$ would produce a cover with
      a larger intersection with $N_{G_1}(r)$, contradicting the choice of
      $\bar C$.
 \end{itemize}

Now, consider the triangle inequality associated with $T$. Since $\bar z^{(1)}$ satisfies it with slack at least $s_1$, there is a multiset $\calC_T\subset\calC^{(1)}$ such that for every $C\in \calC_T$, $|C\cap D(T)|\geq 3$. 
Moreover, the total sum of the slack is at least $Ns_1$, since
\[
\sum_{C\in \calC_T} (|C\cap D(T)|-2) \geq Ns_1.
\]
Suppose that one such set $C\in \calC_T \cap \calC^{(1)}_r$.  Hence, $r\in C$, $(r,u), (r,v)\not\in C$. Notice that, we can assume that $C= \bar C$, hence $u \notin C$. Otherwise apply Lemma \ref{lem:swap} with $h=r$, $H_ 2= B(r,\{u,v\})$, $A = C$ and $B = \bar C$. In particular, since $C \in \calC_{T}$ we must have $(u,v), v\in C$. 
Now, apply Lemma~\ref{lem:swap} once again to $G_1$ with $h=u$ and  $H_2= \overline{B(u,r)}$. 
However, the hypotheses of the lemma may fail to hold for $A=C_u$ and $B=C= \bar C$, since $u\notin B$ and $(u,v) \notin A$. In this case we also move edge $(u,v)$ from $B$ to $A$. Let $A'=A\setminus D(H_2)\cup (B\cap D(H_2))\cup\{(u,v)\}$, and $B'=B\setminus (D(H_2)\cup \{(u,v)\})\cup (A\cap D(H_2))$, then one easily checks that the thesis of Lemma \ref{lem:swap} still holds.  
 Hence, we can assume that $\calC_T\cap \calC^{(1)}_r=\emptyset$. Moreover, we claim that  $(u,v)\not\in C$. Otherwise we have the following cases:

\begin{itemize}
    \item $u, v\in C$, in which case $C$ can afford to lose edge $(u,v)$, a contradiction since no cover in $\calC^{(1)}$, in particular in $\calC_T$, can afford to lose edge $(u,v)$;
    \item $u\not\in C$, hence $(r, u)\in C$ (since $C$ covers the edge and $r\not\in C$) and $v\in C$ (since $C\in \calC_T$), in which case $C$ can afford to lose edge $(u,v)$, again a contradiction;
    \item $v\not\in C$, and we can reach a contradiction analogously.
\end{itemize} 
Hence for any $C\in \calC_T$, $(u,v)\not\in C$, implying that $u, v\in C$. Notice that this also implies that the slack associated with the triangle inequality is 1 for each such $C$, hence $|\calC_T|\geq Ns_1$ and there must be one such $C$ with $C\not\in \calC^{(1)}_s$, i.e. $C\in \calC^{(1)}_0$.
Finally, we are left with two cases:
\begin{itemize}
    \item $C\cap D(T)= \{(r,u), u,v\}$. In this case, we apply Lemma \ref{lem:swap} to $G_1$ with $h=u$, $H_2 = \overline{B(u,r)}$, 
    $A=C$ and $B=\bar C$ and obtain a contradiction similarly as above.
    \item $C\cap D(T)= \{(r,v), u,v\}$. Recall that $C_u$ satisfies $u\in C_u$, $v\not\in C_u$. Then we can apply Lemma \ref{lem:swap} twice. First, with $h=r$, $H_2 = B(r,\{u,v\})$, 
    $A=C_u$, $B=\bar C$ (so that now $\bar C$ misses $v$ instead of $u$). Then, we apply the lemma with $h=v$, $H_2 = \overline{B(v,r)}$, 
    $A=C$, $B=\bar C$. This gives again a contradiction. 
    \end{itemize}
    
\end{claimproof}

The two claims together show that $z$ is in $P_{TC}(G)$. This concludes the proof.
\end{proof}

The previous proof may contain redundant inequalities in the linear description. In view of Propositions \ref{vertex_constr}--\ref{prop:edge_facet} and Theorem \ref{clique_ineq}, we obtain a stronger result by providing a minimal linear description for this latter class of graphs, given by a set of facet-defining inequalities (which are unique since $P_{TC}(G)$ is full-dimensional).
\begin{cor}
Let $G=(V,E)$ be a triangular cactus graph with $|V| \geq 2$. Let $V_2:=\{v \in V \mid \deg(v) \geq 2\}$, $E_{2}:=\{e=(v,w) \in E \mid \deg(v),\deg(w) \geq 2\}$ and $\mathcal{T}_3:=\{T \ \text{triangle of $G$} \mid \deg(v) \geq 3, \forall v \in V(T)\}$. Then, a \textbf{minimal complete} linear description is given by:
\begin{align*}
P_{TC}(G):= \bigg\{(x,y): \quad &x_v+\sum\limits_{e \in \delta(v)}y_e \geq 1, \qquad & \forall v \in V, \\
&x_v + x_w+y_e \geq 1, \qquad & \forall e=(v,w) \in E_2, \\
& \sum\limits_{v \in V(T)}x_v + \sum\limits_{e \in E(T)}y_e \geq 2, \qquad & \forall T \in \mathcal{T}_3, \\
&z_a \geq 0, \qquad & \forall a \in V_2 \cup E_2 ,\\
&z_a \leq 1, \qquad & \forall a \in V \cup E \bigg \}.
\end{align*}
\end{cor}
We stress that triangular cactus graphs are the only class of cactus graphs for which we obtain a complete polyhedral characterization where all the facets have been fully identified. This result cannot be extended to arbitrary cactus graphs. To illustrate this limitation, consider the one-sum of two cycles of length $5$.
\begin{figure}[h]
\begin{center}
\begin{tikzpicture}[
    scale=1.0,
    vertex/.style={
        circle,
        draw,
        fill=black,
        minimum size=3mm,
        inner sep=0pt
    }
]

\foreach \i/\ang in {1/0,2/72,3/144,4/216,5/288}
{\node[vertex,label=\ang:$\i$]
      (L\i) at ({-2+2*cos(\ang)},{2*sin(\ang)}) {};
}
\draw (L1)--(L2)--(L3)--(L4)--(L5)--(L1);

\foreach \i/\ang in {6/72,7/144,8/216,9/288}
{
    \node[vertex,label=\ang:$\i$]
      (R\i) at ({2+2*cos(180-\ang)},{2*sin(180-\ang)}) {};
}
\draw (L1)--(R6)--(R7)--(R8)--(R9)--(L1);

\end{tikzpicture}
\end{center}
\caption{A cactus graph consisting of two cycles of length $5$.}\label{cactus5}
\end{figure}
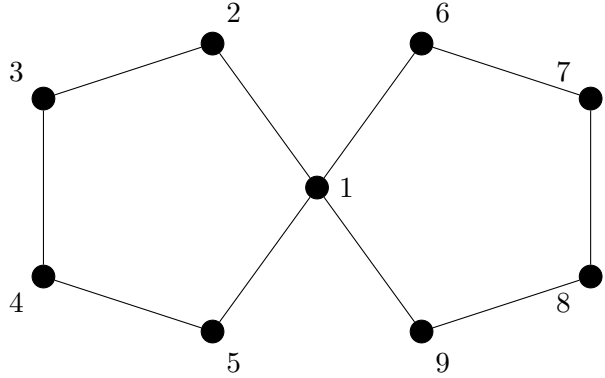
The cactus graph depicted in Figure \ref{cactus5} exhibits a structure that differs from the types of inequalities considered previously. In particular, one can verify that the following inequality is facet-defining for the associated polytope:
\[
x_1 + 2x_2 + x_3 + x_4 + x_5 + y_{12} + y_{15} + y_{23} + y_{34} + y_{45} \geq 4.
\]
Observe that none of the previously identified facets provides a structural or geometric characterization of this inequality. This indicates that the situation becomes more intricate when dealing with cactus graphs of varying sizes.

\medskip

\section{Conclusion}
In this paper, we introduced the STCP and investigated its structural, algorithmic, and polyhedral properties. For the unweighted problem, we related strong total covers to dominating sets
of incidence graphs and proved a Gallai-type identity with maximum induced
matchings. These equivalences yield NP-hardness and identify several graph
classes on which the STCP is polynomial-time solvable.
For the weighted version, we showed that the problem is polynomially solvable on complete bipartite graphs, while it remains NP-hard on complete graphs, even if one restricts the weights to be binary. Moreover, we obtained a primal-dual \(2\)-approximation algorithm, whose analysis
also shows that the basic covering relaxation \eqref{lp_tc} has integrality gap equal to \(2\). Finally, we initiated the polyhedral study of STCP by characterizing some valid and facet-defining inequalities and by providing a complete linear description for trees, cycles and triangular cactus graphs.

Several open questions arise from our results. On the algorithmic side, the duality relationship with the (weighted) total matching problem naturally leads to a demand version of the STCP, where each element \(a\in D\) has a nonnegative integer demand \(d_a\) and must be covered at least \(d_a\) times. How hard is this problem? Although the primal-dual algorithm developed here does not seem to extend to this setting, we believe that related LP-based techniques should attain a constant-factor approximation, possibly with factor 2.

On the polyhedral side, the next natural step would be to understand cactus graphs with longer cycle blocks, where the example above already suggests the appearance of inequalities with larger coefficients. A more structural question is:
how do facets behave under \(1\)-sums? Answering this might also contribute to the study of dominating set polytopes on cactus graphs. At the other extreme, the strong total cover polytope of complete graphs is not yet understood: while our results show that all the subcliques give facet-defining inequalities, computational results show that we are far from obtaining a complete description. Finding further classes of facets, and more generally understanding when inequalities supported on subgraphs remain valid or facet-defining in the ambient graph, are natural next questions.


%


\smallskip 

\noindent  {\bf Acknowledgments.} This research benefited from the support of the FMJH Program Gaspard Monge for optimization and operations research and their interactions with data science, and from the support of the SID project ``Network optimization problems under extra constraints'' under the BIRD 2024 program sponsored by the University of Padua. 

\bibliographystyle{apalike}
\bibliography{references}

\end{document}